\documentclass[a4paper,fleqn]{cas-dc}

\usepackage[numbers,longnamesfirst]{natbib}
\usepackage{amsmath,amssymb,amsfonts}
\usepackage{graphicx}
\usepackage{xcolor}
\usepackage{subfigure}
\usepackage{booktabs}
\usepackage{multirow}
\usepackage{makecell}
\usepackage{threeparttable}
\usepackage{diagbox}
\usepackage{enumitem} % for itemize distance
\usepackage{url}
\usepackage[ruled, linesnumbered]{algorithm2e}
\newtheorem{question}{Question}
\newtheorem{scenario}{Scenario}
\newtheorem{theorem}{Theorem}
\newtheorem{lemma}{Lemma}
\newtheorem{proposition}{Proposition}
\newtheorem{corollary}{Corollary}
\newproof{proof}{Proof}

\begin{document}
	
\title[mode = title]{Equivalence, Identity, and Unitarity Checking in Black-Box Testing of Quantum Programs}
\shorttitle{Equivalence, Identity, and Unitarity Checking in Black-Box Testing of Quantum Programs}

\author[1]{Peixun Long}[orcid=0009-0000-9255-9335]
\affiliation[1]{organization={State Key Laboratory of Computer Science, Institute of Software, Chinese Academy of Science; University of Chinese Academy of Science},
city={Beijing},
country={China}}
\cormark[1] 
\ead{longpx@ios.ac.cn}

\author[2]{Jianjun Zhao}[orcid=0000-0002-2208-0289]
\affiliation[2]{organization={Kyushu University},
city={Fukuoka},
country={Japan}}
\ead{zhao@ait.kyushu-u.ac.jp}

\cortext[cor1]{Corresponding author} 

\begin{abstract}
Quantum programs exhibit inherent non-deterministic behavior, which poses more significant challenges for error discovery compared to classical programs. While several testing methods have been proposed for quantum programs, they often overlook fundamental questions in black-box testing. In this paper, we bridge this gap by presenting three novel algorithms specifically designed to address the challenges of equivalence, identity, and unitarity checking in black-box testing of quantum programs. We also explore optimization techniques for these algorithms, including specialized versions for equivalence and unitarity checking, and provide valuable insights into parameter selection to maximize performance and effectiveness. To evaluate the effectiveness of our proposed methods, we conducted comprehensive experimental evaluations, which demonstrate that our methods can rigorously perform equivalence, identity, and unitarity checking, offering robust support for black-box testing of quantum programs.
\end{abstract}

%\ccsdesc[500]{Software and its engineering~Software testing and debugging}

\maketitle

\section{Introduction}
\label{sec:introduction}
Quantum computing, which utilizes the principles of quantum mechanics to process information and perform computational tasks, is a rapidly growing field with the potential to revolutionize various disciplines~\cite{national2019quantum}. It holds promise for advancements in optimization~\cite{farhi2014quantum}, encryption~\cite{mosca2018cybersecurity}, machine learning~\cite{biamonte2017quantum}, chemistry~\cite{mcardle2018quantum}, and materials science~\cite{yang2017mixed}. Quantum algorithms, when compared to classical algorithms, offer the potential to accelerate the solution of specific problems~\cite{deutsch1985quantum,grover1996fast,shor1999polynomial}. As quantum hardware devices and algorithms continue to develop, the importance of creating high-quality quantum software has become increasingly evident. However, the nature of quantum programs, with their special characteristics such as superposition, entanglement, and non-cloning theorems, makes it challenging to track errors in these programs~\cite{miranskyy2020your,huang2019statistical}. Therefore, effective testing of quantum programs is crucial for the advancement of quantum software development.

Black-box testing~\cite{beizer1995black} is a software testing method that assesses the functionality of a program without examining its internal structure or implementation details. This method has broad applications for identifying software errors and improving software reliability. The inherently non-deterministic behavior of quantum programs makes error detection more challenging than in classical programs. Additionally, due to the potential interference of measurement with quantum states, observing the internal behavior of quantum programs becomes nearly impossible. Consequently, black-box testing assumes a crucial role in testing quantum programs. While several testing methods for quantum programs have been proposed~\cite{abreu2022metamorphic,ali2021assessing,fortunato2022mutation,li2020projection,miranskyy2019testing,wang2018quanfuzz}, these methods have paid limited attention to the fundamental questions in black-box testing of quantum programs. Important questions essential to black-box testing for quantum programs have remained largely unexplored.

Black-box testing of quantum programs refers to testing these programs based solely on selecting the appropriate inputs and detecting the corresponding outputs. To effectively address the challenges associated with black-box testing in quantum programs, it is essential to explore and answer the following fundamental research questions (RQs):

\begin{itemize}
    \item[(1)] \textbf{Equivalence Checking:} Given two quantum programs $\mathcal{P}$ and $\mathcal{P'}$, how can we determine whether they are equivalent?
    
    \item[(2)] \textbf{Identity Checking:} Given a quantum program $\mathcal{P}$, how can we check whether it represents an identity transform?
    
    \item[(3)] \textbf{Unitarity Checking:} Given a quantum program $\mathcal{P}$, how can we ascertain whether it represents a unitary transform?
\end{itemize}

In this paper, we aim to address these RQs and lay the foundation for black-box testing of quantum programs. We propose three novel methods that specifically target equivalence, identity, and unitarity checking in quantum programs. For equivalence checking, we introduce a novel algorithm based on the Swap Test, which compares the outputs of two quantum programs on Pauli input states. To simplify identity checking, we present a straightforward algorithm to avoid repeated running. Additionally, we derive a critical theorem that provides a necessary and sufficient condition for a quantum operation to be a unitary transform, enabling us to develop an effective unitarity checking algorithm. 

Furthermore, we conduct theoretical analyses and experimental evaluations to demonstrate the effectiveness and efficiency of our proposed algorithms. We also discuss the optimization of our algorithms. The results of our evaluations confirm that our algorithms successfully perform equivalence, identity, and unitarity checking, supporting the black-box testing of quantum programs.

In summary, our paper makes the following contributions:

\begin{itemize}
\item \textbf{A theorem on unitarity checking:} We proved a critical theorem about how to check whether a quantum operation is a unitary transform, which is a basis for developing our algorithm for unitarity checking.

\item \textbf{Checking algorithms:} We develop three novel algorithms for checking the equivalence, identity, and unitarity of quantum programs, respectively.

\item \textbf{Algorithm optimization:} We explore the optimization of our checking algorithms by devising optimized versions that specifically target equivalence and unitarity checking. We also discuss and provide detailed insights into the selection of algorithm parameters to maximize their performance and effectiveness.
\end{itemize}

Through these contributions, our paper advances the foundation of black-box testing for quantum programs and provides valuable insights for quantum software development.

The organization of this paper is as follows. Section~\ref{sec:background} introduces some basic concepts and technologies of quantum computation.
We discuss questions and their motivations we want to address in this paper in Section~\ref{sec:questions}. Section~\ref{sec:methods} presents novel algorithms to solve these questions.
Section~\ref{sec:optimization} discusses the optimization of equivalence checking and unitarity checking.
Section~\ref{sec:Experments} discusses the experimental evaluation. Section~\ref{sec:threats} discusses the threats to the validity of our methods. We discuss related work in Section~\ref{sec:relatedwork}, and the conclusion is given in Section~\ref{sec:conclusion}.

\section{Background}
\label{sec:background}
We next introduce some background knowledge about quantum computation which is necessary for understanding the content of this paper. 

\subsection{Basic Concepts of Quantum Computation}
\label{subsec:basic_concepts}

%\textbf{Qubit and quantum state.}
A \textit{qubit} is the fundamental unit of quantum computation. Like the classical bit has values 0 and 1, a qubit also has two \textit{basis states} with the form $\left|0\right>$ and $\left|1\right>$. However, a qubit is allowed to contain a \textit{superposition} between basis states, with the general state being a linear combination of $\left|0\right>$ and $\left|1\right>$: $a\left|0\right>+b\left|1\right>$, where $a$ and $b$ are two complex numbers called \textit{amplitudes} that satisfy $|a|^2+|b|^2=1$. The amplitudes describe the proportion of $\left|0\right>$ and $\left|1\right>$. For multiple qubits, the basis states are analogous to binary strings. For example, a two-qubit system has four basis states: $\left|00\right>$, $\left|01\right>$, $\left|10\right>$ and $\left|11\right>$, and the general state is $a_{00}\left|00\right>+a_{01}\left|01\right>+a_{10}\left|10\right>+a_{11}\left|11\right>$, where $|a_{00}|^2+|a_{01}|^2+|a_{10}|^2+|a_{11}|^2=1$. The state can also be written as a column vector $[a_{00},a_{01},a_{10},a_{11}]^T$, which is called \textit{state vector}. Generally, the quantum state on $n$ qubits can be represented as a state vector $\left|\psi\right>$ in \textit{Hilbert Space} of $d$ dimension, where $d=2^n$.

%\textbf{Pure state and mixed state.}
The states mentioned above are all \textit{pure states}, which means they are not probabilistic. Sometimes a quantum system may have a probability distribution over several pure states rather than a certain state, which we call a \textit{mixed state}. Suppose a quantum system is in state $\left|\psi_i\right>$ with probability $p_i$. The completeness of probability gives that $\sum_i p_i = 1$. We can denote the mixed state as the \textit{ensemble representation}: $\{(p_i, \left|\psi_i\right>)\}$.

%\textbf{Density operator.}
Besides state vectors, the \textit{density matrix} or \textit{density operator} is another way to express a quantum state, which is convenient for expressing mixed states. Suppose a mixed state has the ensemble representation $\{(p_i, \left|\psi_i\right>)\}$, the density matrix of this state is $\rho = \sum_i p_i \left|\psi_i\right>\left<\psi_i\right|$, where $\left<\psi_i\right|$ is the conjugate transpose of $\left|\psi_i\right>$ (thus it is a row vector and $\left|\psi_i\right>\left<\psi_i\right|$ is a $d\times d$ matrix). Owing that the probability $p_i\geq 0$ and $\sum_i p_i = 1$, a lawful and complete density matrix $\rho$ should satisfy: (1) $tr(\rho) = 1$ and (2) $\rho$ is a positive matrix\footnote{For any column vector $\left|\alpha\right>$, $\left<\alpha\right|\rho\left|\alpha\right> \geq 0$}. $\rho$ represents a pure state if and only if $tr(\rho^2) = 1$~\cite{nielsen2002quantum}. Obviously, The density matrix of a pure state $\left|\phi\right>$ can be written as $\left|\phi\right>\left<\phi\right|$. In this paper, we will use both state vector and density matrix representations.

\subsection{Evolution of Quantum States}
\label{subsec:quantum_op}

Quantum computing is performed by applying proper \textit{quantum gates} on qubits. An $n$-qubit quantum gate can be represented by a $2^n \times 2^n$ unitary matrix $G$. Applying gate $G$ on a state vector $\left|\psi\right>$ will obtain state vector $G\left|\psi\right>$. Moreover, applying $G$ on a density matrix $\rho$ will obtain density matrix $G \rho G^{\dagger}$, where $G^{\dagger}$ is the \textit{conjugate transpose} of $G$. For a unitary transform, $G^\dagger = G^{-1}$, where $G^{-1}$ is the \textit{inverse} of $G$.

There are several basic quantum gates, such as single-qubit gates $X$, $Z$, $S$, $H$, and two-qubit gate CNOT. The matrices of them are:

\begin{small}
$$
X=\left[\begin{array}{cc}
0 & 1\\
1 & 0
\end{array}
\right], \quad
Z=\left[\begin{array}{cc}
1 & 0\\
0 & -1
\end{array}
\right], \quad
S=\left[\begin{array}{cc}
1 & 0\\
0 & i
\end{array}
\right],
$$
$$
H=\frac{1}{\sqrt{2}} \left[\begin{array}{cc}
1 & 1\\
1 & -1
\end{array}
\right], \quad
\mathrm{CNOT}=\left[\begin{array}{cccc}
1&0&0&0\\
0&1&0&0\\
0&0&0&1\\
0&0&1&0
\end{array}
\right]
$$
\end{small}

In quantum devices, the information of qubits can only be obtained by \textit{measurement}. Measuring a quantum system will get a classical value with the probability of corresponding amplitude. Then the state of the quantum system will collapse into a basis state according to obtained value. For example, measuring a qubit $a\left|0\right>+b\left|1\right>$ will obtain result 0 and collapse into state $\left|0\right>$ with probability $|a|^2$; and obtain result '1' and collapse into state $\left|1\right>$ with probability $|b|^2$. This property brings uncertainty and influences the testing of quantum programs.

A quantum circuit is a popular model to express the process of quantum computing. Every line represents a qubit in a quantum circuit model, and a sequence of operations is applied from left to right. Figures~\ref{fig:swap} and \ref{fig:swaptest} are two examples of quantum circuits.

%Quantum gates represent unitary transforms, and the combination of several quantum gates is still a unitary transform. However, a unitary transform can only represent the evolution of a closed quantum system, not the evolution of a general or random system. For example, a measurement cannot be represented by any unitary transform.

\subsection{Quantum Operations and Quantum Programs}

%The combinations of quantum gates represent unitary transforms, but measurements are not unitary transforms.

If a quantum program is composed only by quantum gates, it will perform a unitary tranform on the input state. However, measurements will destory the unitarity.

\textit{Quantum operation} is a general mathematical model that can describe the general evolution of a quantum system~\cite{nielsen2002quantum}. It is a linear map on the space of the density matrix. In fact, both quantum gates and measurements can be represented by quantum operations. Suppose we have an input density matrix $\rho$ ($d \times d$ matrix). After an evolution through quantum operation $\mathcal{E}$, the density matrix becomes $\mathcal{E}(\rho)$. $\mathcal{E}(\rho)$ can be represented as the \textit{operator-sum representation}~\cite{nielsen2002quantum}: $\mathcal{E}(\rho) = \sum_i E_i \rho E_i^{\dagger}$, where $\{E_i\}$ is a group of $d \times d$ matrices. If $\sum_i E_i^{\dagger} E_i = I$, then $\mathcal{E}$ is called a \textit{trace-preserving} quantum operation.

In developing quantum programs, programmers use quantum gates, measurements, and control statements to implement some quantum algorithms. A quantum program usually transforms the quantum input state into another state. An important fact is that a quantum program with \textit{if} statements and \textit{while-loop} statements, where the \textit{if} and \textit{while-loop} conditions can contain the result of measuring qubits, can be represented by a quantum operation~\cite{ying2016foundationQP}. So, quantum operations formalism is a powerful tool that can be used by our testing methods to model the state transformations of quantum programs under test. In addition, a program will eventually terminate if and only if it corresponds to a trace-preserving quantum operation~\cite{ying2016foundationQP}. This paper focuses only on trace-preserving quantum operations, i.e., the quantum programs that always terminate.

\begin{figure}
\centering
    \includegraphics[scale=0.6]{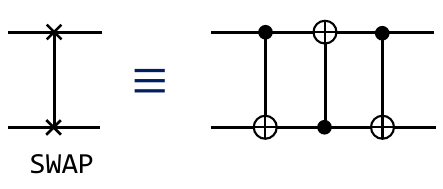}
    \caption{A SWAP gate that can be implemented by three CNOT gates.}
    \label{fig:swap}
    \centering
    \includegraphics[scale=0.6]{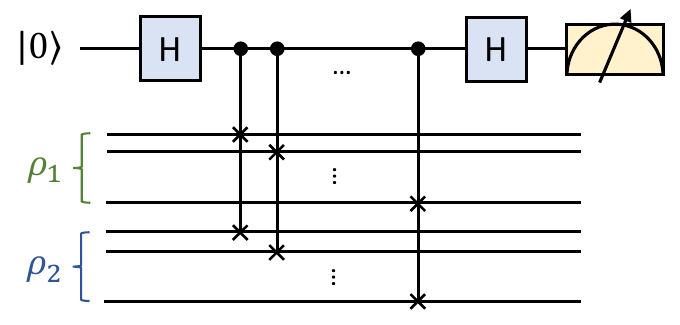}
    \caption{The quantum circuit of Swap Test.}
    \label{fig:swaptest}
\end{figure}

\subsection{Swap Test}
\label{subsec:swap_test}

The Swap Test~\cite{buhrman2001quantum,barenco1997stabilization} is a procedure in quantum computation that allows us to determine how similar two quantum states are. Figure~\ref{fig:swaptest} shows the quantum circuit of the Swap Test. It includes two (single or multi-qubit) registers carrying the states $\rho_1$ and $\rho_2$ and a single ancillary qubit initialized in state $\left|0\right>$.
It contains a series of "Controlled-SWAP" gates on each pair of corresponding qubits of target states $\rho_1$ and $\rho_2$, and the ancilla qubit is used for controlling. SWAP gate can be implemented as three CNOT gates as shown in Figure~\ref{fig:swap}. There are two $H$ gates on the ancilla qubit before and after these Controlled-SWAP gates, respectively. The ancilla qubit is measured at the end, and the probability that result '1' occurs, denoted as $p_1$, is related to states $\rho_1$ and $\rho_2$ as shown in the following formula (\ref{equ:SwapTest}):

\begin{equation}
\label{equ:SwapTest}
tr(\rho_1\rho_2) = 1-2p_1
\end{equation}

Note that the probability $p_1$ can be estimated by repeating running the Swap Test and counting the proportion of obtaining result '1'. Based on this process, we can estimate parameter $tr(\rho_1\rho_2)$ and then obtain more useful information on states $\rho_1$ and $\rho_2$. 
In particular, if we set $\rho_1 = \rho_2 = \rho$, we can estimate $tr(\rho^2)$ and use it to determine whether $\rho$ is a pure or mixed state. Also, if we let $\rho_1$ and $\rho_2$ be two pure states, where $\rho_1 = \left|\alpha\right>\left<\alpha\right|$ and $\rho_2 = \left|\beta\right>\left<\beta\right|$, then we have $tr(\rho_1\rho_2) = \left|\left<\alpha|\beta\right>\right|^2$, where $\left<\alpha|\beta\right>$ represents the \textit{inner product} of two state $\left|\alpha\right>$ and $\left|\beta\right>$. We call $\left|\alpha\right>$ and $\left|\beta\right>$ are \textit{orthogonal}, denoted as $\left|\alpha\right> \bot \left|\beta\right>$, if $\left<\alpha|\beta\right> = 0$.

\subsection{Quantum Tomography}
\label{subsec:quantum_tomo}

Quantum tomography~\cite{d2003quantum} is used for obtaining the details of a quantum state or operation. Owing that measurement may collapse a quantum state, we need many copies of the target quantum state or repeat the target quantum operation many times to reconstruct the target. \textit{State tomography} reconstructs the information of a quantum state, and \textit{process tomography} reconstructs the information of a quantum operation.

In quantum tomography, The following four \textit{Pauli matrices} are important:

\begin{small}
$$
\label{equ:paulis}
\sigma_0 = \left[
\begin{array}{cc}
	1 & 0\\
	0 & 1
\end{array}
\right], \qquad
\sigma_1 = \left[
\begin{array}{cc}
	0 & 1\\
	1 & 0
\end{array}
\right],
$$
$$
\sigma_2 = \left[
\begin{array}{cc}
	0 & -i\\
	i & 0
\end{array}
\right], \qquad
\sigma_3 = \left[
\begin{array}{cc}
	1 & 0\\
	0 & -1
\end{array}
\right]
$$
\end{small}

\noindent 
All $\sigma_1$, $\sigma_2$ and $\sigma_3$ have two eigenvalues: $-1$ and $1$. $\sigma_3$ has eigenstates $\left|0\right> = [1, 0]^T$ (for eigenvalue $1$) and $\left|1\right> = [0, 1]^T$ (for eigenvalue $-1$). $\sigma_1$ has eigenstates $\left|+\right> = \frac{1}{\sqrt 2}(\left|0\right>+\left|1\right>)$ (for $1$) and $\left|-\right> = \frac{1}{\sqrt 2}(\left|0\right>-\left|1\right>)$ (for $-1$). $\sigma_2$ has eigenstates $\left|+_i\right> = \frac{1}{\sqrt 2}(\left|0\right>+i\left|1\right>)$ (for $1$) and $\left|-_i\right> = \frac{1}{\sqrt 2}(\left|0\right>-i\left|1\right>)$ (for $-1$). $\sigma_0$ has only one eigenvalue $1$, and any single-qubit state is its eigenstate.

An important fact is that the four Pauli matrices form a group of bases of the space of the single-qubit density matrix. Similarly, the tensor product of Pauli matrices $\sigma_{\vec{v}} = \sigma_{v_1}\otimes \sigma_{v_2} \otimes \cdots \otimes \sigma_{v_n}$ form a basis of an $n$-qubit density matrix, where $v_i \in \{0,1,2,3\}$, $\vec{v} = (v_1,\dots,v_n)$.
So a density matrix $\rho$ can be represented as the linear combination of Pauli basis~\cite{nielsen2002quantum}:

\begin{equation}
\label{equ:RhoDecomp}
\rho = \sum_{\vec{v}} \frac{tr(\sigma_{\vec{v}}\rho)}{2^n}\sigma_{\vec{v}},
\end{equation}

\noindent where the combination coefficient for basis matrix $\sigma_{\vec{v}}$ is $tr(\sigma_{\vec{v}}\rho)/2^n$. In practice, $tr(\sigma_{\vec{v}}\rho)$ can be estimated by \textit{Pauli measurement} of Pauli matrix $\sigma_{\vec{v}}$. We can conduct experiments to observe each $tr(\sigma_{\vec{v}}\rho)$ on $\rho$ and then use formula (\ref{equ:RhoDecomp}) to reconstruct the density matrix of $\rho$.

A quantum operation $\mathcal{E}$ is a linear transformation on the input space, so it depends on the behavior of a group of bases. We can use state tomography to reconstruct the corresponding output density matrices of all basis input states, and then $\mathcal{E}$ can be reconstructed by these output density matrices~\cite{nielsen2002quantum}.

\section{Questions and Motivations}
\label{sec:questions}
We next outline the key questions addressed in this paper and provide a comprehensive rationale for each question, supported by relevant scenarios. By doing so, we aim to highlight the motivation and necessity behind investigating these specific aspects.

\subsection{Modelling Black-Box Quantum Programs}
\label{subsec:model}

In this paper, we focus on black-box testing, where we do not require prior knowledge of the internal structure of the target program. Instead, we achieve the testing objective by observing the program's outputs when provided with suitable inputs. This approach allows us to evaluate the program's behavior without relying on specific implementation details.

The beginning of our discussion is to model a black-box quantum program. In~\cite{ali2021assessing}, Ali~\textit{et al.} gives some definitions of a quantum program, input, output, and program specification to guide the generation of test cases. They define the program specification as the expected probability distribution of the classical output values under the given inputs. As~\cite{long2022process} points out, this definition is inherently classical, implying that the program ends with measurements to transform quantum information into classical probabilities. It is more reasonable to deal directly with quantum information rather than transform them into classical probabilities. Accordingly, the quantum input and output of the program can be modeled by quantum states (state vectors or density matrices).

According to Section~\ref{subsec:quantum_op}, a quantum program with \textit{if} statements and \textit{while-loop} statements can be represented by a quantum operation~\cite{ying2016foundationQP}. So in this paper, we use the quantum operation model to represent a black-box quantum program. A black-box quantum program with quantum input variable $\rho$ can be represented by an unknown quantum operation $\mathcal{E}$, and $\mathcal{E}(\rho)$ represents the corresponding output under input $\rho$. The research on the properties of quantum programs can be transformed into research on quantum operations.

\subsection{Questions}

We introduce the key questions addressed in this paper, namely \textit{equivalence checking}, \textit{identity checking}, and \textit{unitarity checking}. While a quantum program can involve classical inputs and outputs, our focus in this paper is solely on the quantum aspects. Specifically, we aim to examine the checking of quantum inputs and outputs while keeping other classical parameters fixed. We outline the three questions addressed in this paper as follows:

\begin{question}
\textbf{Equivalence Checking}

Given two quantum programs $\mathcal{P}_1$ and $\mathcal{P}_2$, how can we determine whether they are equivalent, meaning they produce the same quantum output for the same quantum input?
\end{question}

\begin{question}
\textbf{Identity Checking}
Given a quantum program $\mathcal{P}$, how can we check whether it represents an identity transform where the input qubits remain unchanged?
\end{question}

\begin{question}
\textbf{Unitarity Checking}

Given a quantum program $\mathcal{P}$, how can we determine whether $\mathcal{P}$ represents a unitary transform?
\end{question}

\subsection{Motivations}

Next, we discuss the motivation and necessity of each question. We will present several scenarios highlighting the practical application of equivalence, identity, and unitary checking in the context of black-boxing testing of quantum programs. We consider these three questions together because they involve typical relations and share similar problem-solving frameworks. Identity checking can be seen as a special case of equivalence checking, with one program being the target $\mathcal{P}$ and the other being the identity operator $I$. However, we also examine identity checking independently for two reasons: firstly, identity relations are among the most common, and secondly, it may be possible to develop a more efficient algorithm for identity checking compared to equivalence checking.

The equivalence of two programs holds substantial importance, as it is a prevalent relationship. Even when testing classical programs, checking the equivalence of two programs is a fundamental testing task. For instance, in the following Scenario~\ref{scn:update}, one of the typical applications of equivalence checking is to guarantee the correctness of an updated program version. While white-box checking for quantum circuit equivalence has been explored in several existing studies~\cite{viamontes2007checking,yamashita2010fast,burgholzer2020advanced,hong2021approximate,hong2022eqchk_dynamic,wang2022eqchk_seq}, our focus in this paper is on black-box checking.

\begin{scenario}
\label{scn:update}
\textsf{Correctness Assurance in Version Update}

One important application of equivalence checking is to ensure the correctness of an updated version of a quantum program in regression testing. Let us consider a scenario where we have an original program $\mathcal{P}$, and after a few months, we develop an updated version $\mathcal{P'}$ by making improvements such as optimizing the program's structure to enhance its execution speed. In this case, ensuring that these modifications do not alter the program's behavior becomes crucial, meaning that $\mathcal{P'}$ should be equivalent to $\mathcal{P}$. By performing an \textbf{equivalence checking}, we can verify if the updated version $\mathcal{P'}$ maintains the same functionality as the original version $\mathcal{P}$.

\end{scenario}

The classical counterpart to identity checking, which involves verifying whether a classical program reproduces its input, lacks practical significance. Identity checking finds significance in the context of quantum programs because creating a \textit{inverse variant} of the original quantum program is a common practice in quantum programming. Many quantum programs inherently involve performing a unitary transformation on the quantum input state. As detailed in Section~\ref{subsec:quantum_op}, any unitary transformation $U$ possesses an inverse denoted as $U^{-1}$, and they satisfy the relation $UU^{-1}=I$, where $I$ denotes the identity transformation. Scenario~\ref{scn:inverse} illustrates how identity checking can be employed to test the inverse variant of the original quantum program.

\begin{scenario}
\label{scn:inverse}
\textsf{Testing the Inverse Variant}

Suppose we have completed a testing task for an original program \texttt{P}. Now, we proceed to implement the reverse version \texttt{invP} of \texttt{P}. While some quantum programming languages, such as Q\#\cite{svore2018q} and isQ\cite{Guo2022isQ}, offer mechanisms for generating \texttt{invP} automatically, these mechanisms come with some restrictions on the target program. For instance, if \texttt{P} contains \textit{if} statements associated with classical input parameters, the quantum behavior of \texttt{P} is a unitary transform for any fixed classical parameters, implying the existence of the inverse \texttt{invP}. 
However, these languages may not have the ability to automatically generate \texttt{invP} as their mechanisms are limited to dealing with simple programs. In such cases, manual implementation is required. As a result, effective testing becomes crucial during the manual implementation process.

To test \texttt{invP}, as proposed in \cite{long2022process}, we only need to check the following relation:

\begin{center}
$\texttt{P} \circ \texttt{invP} = I$
\end{center}

\noindent where `$\circ$' represents the sequential execution (from right to left) of the two subroutines. This task involves \textbf{checking the identity relationship} between the composite program and the identity program ($I$).

\end{scenario}

Besides the inverse variant, as mentioned in \cite{long2022process}, there are also two other variants of an original program \texttt{P} - \textit{controlled variants} \texttt{CtrlP} and \textit{power variants} \texttt{PowP}. For example, the testing process for power variants can be reduced to check the following relations:

\begin{center}
For $k>0$,\quad $\texttt{InvP}^k \circ \texttt{PowP}(k) = I$;

For $k<0$,\quad $\texttt{P}^{|k|} \circ \texttt{PowP}(k) = I$.
\end{center}

%\red{It seems that identity checking is a special case of equivalence checking - one of the programs is fixed into identity.} However, it is valuable to consider this question separately. First, compared with general equivalence checking, there is a simpler and faster method to solve identity checking. Second, identity is one of the most common relations. For example, \cite{long2022process} gives some methods to execute testing tasks for variant subroutines after finishing the testing of the original program. Consider the scenario of testing the inverse variant of a quantum program:

%In addition, identity checking can be extended to check whether \texttt{P} represents a known unitary transform $U$, as the following scenario shows.

%\vspace{2mm}
%\noindent\textit{\underline{Scenario 4.}} To check whether \texttt{P} represents a known unitary transform $U$, we just need to check the relation: $U^{-1}\circ\mathcal{P} = I$.
%\vspace{2mm}
%This question is different from Scenario 4: the unitary transform $U$ in Scenario 4 is known, whereas, in unitarity checking, it is unknown.

The motivation for unitarity checking is based on the fact that many practical quantum programs are unitary transformations, which means that they do not contain measurements. If the unitarity checking fails for these programs, we can know that there exist some unexpected measurements in them. In fact, unitarity checking reflects the unique properties of quantum programs, and there is no classical counterpart to this type of checking. Scenario~\ref{scn:specification} gives another case where we need to employ unitarity checking.

\begin{scenario}
\label{scn:specification}
\textsf{Checking the Program Specification}

Checking whether the program output conforms to the program specification is a crucial step in testing. However, unlike classical programs, where the output can be directly observed, checking the output of quantum programs under arbitrary inputs can be challenging. Fortunately, as discussed in \cite{long2022process}, if the intended behavior of the program is to perform a unitary transform, we can simplify the checking process. It involves:

\begin{itemize}
\item[(1)] Checking the output under classical input states;
\item[(2)] Performing additional \textbf{unitarity checking} on the program.
\end{itemize}
\end{scenario}

In practice, it may appear that testing all possible quantum inputs is necessary, yet the number of quantum inputs is infinite. Fortunately, the linearity of quantum operations enables us to examine only a finite subset of all potential inputs. Moreover, our objective is testing rather than rigorous mathematical proof. Consequently, we only need to select a small, appropriate set of test cases to validate.

%\cite{long2022process} proposes that we need to test both classical and superposition input for a quantum program. However, the corresponding output on superposition input is usually difficult to check. Unitary transform has a good property that is determined by the behavior of a group of bases, and all classical states form a group of common bases. If we can perform unitarity checking, the test process for unitarity programs can be improved.

%\Jianjun{Which title is more suitable for the title of Section 5: Testing or Checking method?}
\section{Checking Algorithms}
\label{sec:methods}
%【就像reviewer-3建议的，我们能否结合算法来描述我们的方法，而不是将描述和算法分开。SE和PL的论文中很多的算法都需要再论文中做详细的说明，这样读起来呀比较容易。】

We next introduce concrete checking algorithms to solve these questions mentioned in Section~\ref{sec:questions}. We discuss the theoretical foundations for each algorithm first and then give concrete algorithms. We denote $n$ as the number of qubits of target programs. According to the discussion in Section~\ref{subsec:model}, we use quantum operation as the equivalent of a quantum program in the following sections.

\subsection{Implementation of Swap Test}
\label{subsec:ImplST}

First, we discuss the implementation of the Swap Test, as it is the fundamental of checking algorithms. The quantum circuit for the Swap Test is shown in Figure~\ref{fig:swaptest}. According to the formula (\ref{equ:SwapTest}), a useful parameter is the probability of returning result '1'. So our implementation returns the number of occurrences of result '1', with a given number of repeat $s$ of the Swap Test. The initial states $\rho_1$ and $\rho_2$ are generated by two input subroutines $\mathcal{P}_1$ and $\mathcal{P}_2$.

The implementation of the Swap Test is shown in Algorithm~\ref{Alg:SwapTest}. Lines $3\sim 4$ initialize quantum variables, where $\left|0\right>^{\otimes n}$ means $n$-qubit all-zero state. Line 5 prepare the input states by $\mathcal{P}_1$ and $\mathcal{P}_2$ on quantum variables \textbf{qs1} and \textbf{qs2}, respectively. Lines $6\sim 8$ implement a series of controlled-SWAP operations in Figure~\ref{fig:swaptest}. Line $10$ measures the ancilla qubit, and lines $11\sim 13$ accumulate the number of results 1. Owing that there are two layers of nested loops (lines $2$ and $6$) and the iteration numbers are $s$ and $n$, respectively, the time complexity of Algorithm~\ref{Alg:SwapTest} is $O(ns)$. As line $3$, it requires $2n+1$ qubits.

\begin{algorithm}
\small
\caption{\texttt{SwapTest}}
\label{Alg:SwapTest}

\KwIn{($n$, $s$, $\mathcal{P}_1$, $\mathcal{P}_2$) \\ \quad $n$: the number of qubits of one target state; \\ \quad $s$: the number of rounds; \\ \quad $\mathcal{P}_1, \mathcal{P}_2$: two subroutines to generate two target states.}
\KwOut{$s_1$: Number of occurrences of result '1'.}

$s_1 \leftarrow 0$\;
\For{i = 1 to $s$}{
	\textbf{qanc} $\leftarrow \left|0\right>$; \textbf{qs1} $\leftarrow \left|0\right>^{\otimes n}$; \textbf{qs2} $\leftarrow \left|0\right>^{\otimes n}$\;
	$H(\textbf{qanc})$\;
	$\mathcal{P}_1(\textbf{qs1})$; $\mathcal{P}_2(\textbf{qs2})$\;
	\For{j = 0 to $n-1$}
	{
		Controlled SWAP(\textbf{qanc}, (\textbf{qs1}[j], \textbf{qs2}[j]))\;
	}
	$H(\textbf{qanc})$\;
	$m \leftarrow$ Measure(\textbf{qanc})\;
	\If{$m = 1$}{$s_1 \leftarrow s_1 + 1$\;}
}
\Return{$s_1$}\;
\end{algorithm}

\subsection{Equivalence Checking}
\label{subsec:EqCheck}

\subsubsection{Theoretical foundations}
A straightforward idea to perform the equivalence checking is to use the quantum process tomography~\cite{d2003quantum} (see Section~\ref{subsec:quantum_tomo}) for two target programs and compare the reconstruction results of these programs. However, quantum process tomography is costly and may contain much unnecessary, redundant information for equivalence checking. However, we can construct the equivalence checking algorithm based on the idea of quantum process tomography - the behavior of a quantum operation $\mathcal{P}$ depends on its behavior under a group of bases of the input density matrix space.

In practical, a common choice of the bases is the Pauli bases $\{\sigma_{\vec{v}}\}$, where $\sigma_{\vec{v}} = \sigma_{v_1}\otimes \sigma_{v_2} \otimes \cdots \otimes \sigma_{v_n}$ is the tensor product of Pauli matrices, $v_i \in \{0,1,2,3\}$, and $\vec{v} = (v_1,\dots,v_n)$. In other words, $\mathcal{P}(\rho)$, the output of quantum operation $\mathcal{P}$ under input $\rho$, depends on every $\mathcal{P} (\sigma_{\vec{v}})$, and $\sigma_{\vec{v}}$ can be further decomposed into the sum of its eigenstates $\sigma_{\vec{v}} = \sum_i{\lambda_i^{\vec{v}}\left|\psi_i^{\vec{v}}\right>\left<\psi_i^{\vec{v}}\right|}$, where $\lambda_i^{\vec{v}}$ is the $i$-th eigenvalue of $\sigma_{\vec{v}}$ and $\left|\psi_i^{\vec{v}}\right>$ is the corresponding eigenstate. So the behavior of a quantum operation depends on the output on every input $\left|\psi_i^{\vec{v}}\right>$. Specifically, $\left|\psi_i^{\vec{v}}\right>$ is the tensor product of single-qubit Pauli eigenstate\footnote{Any single-qubit state is the eigenstate of $\sigma_0$, so we only need to consider the eigenstates of $\sigma_1$, $\sigma_2$, and $\sigma_3$.}:

\begin{equation}
\label{equ:multipauli}
\left|\psi_i^{\vec{v}}\right> \in \{ \left|0\right>, \left|1\right>, \left|+\right>, \left|-\right>, \left|+_i\right>, \left|-_i\right> \} ^ {\otimes n} = A
\end{equation}

\noindent where states $\left|+\right>$, $\left|-\right>$, $\left|+_i\right>$ and $\left|-_i\right>$ are defined in Section~\ref{subsec:quantum_tomo}. Thus, the behavior of a quantum operation depends on the output under every input state in set $A$. We call the states in set $A$ as \textit{Pauli input states}. Note that no matter what Pauli index vector $\vec{v}$ is, state $\left|\psi_i^{\vec{v}}\right>$ has the form of formula (\ref{equ:multipauli}). It deduces that two quantum operations are equivalent if and only if their outputs are equivalent on all $6^n$ Pauli input states\footnote{Actually, the degree of freedom for an $n$-qubit density matrix is $4^n$, indicating that some Pauli states lack independence. Nevertheless, this is inconsequential for our methods, as they only require sampling a small subset of Pauli input states.}.

Commonly, the equivalent of two states $\rho_1$ and $\rho_2$ can be judged by \textit{trace distance} $D(\rho_1,\rho_2) = \frac{1}{2}tr|\rho_1-\rho_2|$ or \textit{fidelity} $F(\rho_1,\rho_2) = tr\sqrt{\sqrt{\rho_1}\rho_2\sqrt{\rho_1}}$~\cite{nielsen2002quantum}. $\rho_1 = \rho_2$ if and only if $D(\rho_1,\rho_2)=0$ or $F(\rho_1,\rho_2)=1$. However, these two parameters are complex to estimate in practice. For a testing task, we need a parameter that is experimentally easy to be obtained. Consider the \textit{Cauchy Inequality} of density matrix $tr(\rho_1\rho_2) \leq \sqrt{tr(\rho_1^2) tr(\rho_2^2)} $, and the \textit{Mean Value Inequation} $\sqrt{xy} \leq \frac{x+y}{2}$, we have

\begin{equation}
\label{equ:Cauchy}
tr(\rho_1\rho_2) \leq \frac{tr(\rho_1^2) + tr(\rho_2^2)}{2}
\end{equation}

\noindent with equality if and only if $\rho_1=\rho_2$. We can construct a parameter:

\begin{equation}
\label{equ:E}
 E(\rho_1,\rho_2) = \left|\frac{tr(\rho_1^2)+tr(\rho_2^2)}{2} - tr(\rho_1\rho_2)\right| 
 \end{equation}

\noindent If $\rho_1=\rho_2$, then $E(\rho_1,\rho_2) = 0$; otherwise $E(\rho_1,\rho_2) > 0$. So given two quantum programs $\mathcal{P}_1$ and $\mathcal{P}_2$, we can estimate $E( \mathcal{P}_1 (\left|\psi_i\right>\left<\psi_i\right| ),$ $ \mathcal{P}_2(\left|\psi_i\right>\left<\psi_i\right| ) )$, where $\left|\psi_i\right>$ is taken over all Pauli eigenstates. To estimate $E(\rho_1, \rho_2)$, we need to estimate $tr(\rho_1\rho_2)$, $tr(\rho_1^2)$ and $tr(\rho_2^2)$, and they can be finished by Swap Test (see Section~\ref{subsec:swap_test}). Suppose we repeat $s$ times for input pairs $(\rho_1, \rho_1)$, $(\rho_2, \rho_2)$, $(\rho_1, \rho_2)$, and there exist $s_1$, $s_2$, and $s_{12}$ times obtaining result '1', respectively. According to formula (\ref{equ:SwapTest}), $tr(\rho_1^2) \approx 1 - \frac{2 s_{1}}{s}$, $tr(\rho_2^2) \approx 1 - \frac{2 s_{2}}{s}$, and $ tr(\rho_1\rho_2) \approx 1 - \frac{2 s_{12}}{s}$. By substituting them into the formula (\ref{equ:E}), we have

\begin{equation}
\label{equ:Et1t2t12}
E(\rho_1, \rho_2) \approx \left|\frac{2s_{12}-s_1-s_2 }{s} \right| = \tilde{E} 
\end{equation}

\noindent
Denote the experimental estimated value of $E(\rho_1, \rho_2)$ as $\tilde{E}$. It seems that we can simply compare $\tilde{E}$ with 0. However, since there will be some errors in the experiment, a better solution is to give a tolerance error $\epsilon$. If $\tilde{E} \leq \epsilon$, which means $\tilde{E}$ is close to 0, we think that $\rho_1$ and $\rho_2$ are equivalent.

There are $6^n$ Pauli eigenstates. %It appears that we need to check an exponential number of input pairs to ensure that a strong conclusion is drawn. 
Fortunately, in program testing tasks, we can tolerate small errors, so we just need to test only a small fraction of the input pairs instead of all of them, i.e., only $k$ input states, where $k \ll 6^n$. We call the checking process for each input a \textit{test point}. For each input state $\left|\psi_i\right>$, we estimate $\tilde{E}$ and return FAIL whenever there is at least one case such that $\tilde{E} > \epsilon$. PASS is returned only when all test points satisfy $\tilde{E} \leq \epsilon$.

\subsubsection{Algorithm}
According to the above discussion, given two target quantum programs $\mathcal{P}_1$ and $\mathcal{P}_2$, the idea of equivalence checking algorithm includes the following three steps:

\begin{itemize}
\item[(1)] Randomly select and generate a set of Pauli input states.
\item[(2)] For each input state $\left|\psi\right>$, run $\mathcal{P}_1$ and $\mathcal{P}_2$, and use Swap Test to check the equivalence of their output states.
\item[(3)] As long as one output pair is not equivalent, return FAIL; otherwise, if all output pairs are equivalent, return PASS.
\end{itemize}

In practice, we can represent each Pauli input state using an $n$-element integer array, where each integer ranges from $0$ to $5$ and corresponds to the single-qubit states defined in formula (\ref{equ:multipauli}). For instance, if $n=3$, the array $K=\left[1,2,3\right]$ represents the input state $\left|\psi_K\right>=\left|1\right>\left|+\right>\left|-\right>$. Constructing the unitary transform $G_K$ to generate state $\left|\psi_K\right>$ from the all-zero state $\left|000\right>$ is straightforward: it involves applying the gates $X$, $H$, and $HX$\footnote{Applying $X$ gate first and then applying $H$ gate.} to the 1st, 2nd, and 3rd qubits, respectively. In fact, the six single-qubit Pauli eigenstates in formula~(\ref{equ:multipauli}) can be generated by applying the gates $I$, $X$, $H$, $HX$, $SH$, and $S^{-1}H$ to the state $\left|0\right>$, respectively.

The algorithm for equivalence checking is presented in Algorithm~\ref{Alg:EqCheck}. We execute $k$ test points (testing $k$ input states, line 1). In each iteration of the loop, a random integer array $K$ for the input state is generated (line 2), and the corresponding unitary transform to generate Pauli state $\left|\psi_K\right>$ is denoted as $G_K$. Notably, the \texttt{SwapTest} subroutine requires two parameters related to generation subroutines for identifying the target states. $\mathcal{P}_1 \circ G_K$ represents a composite subroutine that executes $G_K$ and $\mathcal{P}_1$ successively, thereby generating the state that results from applying $\mathcal{P}_1$ to the input produced by $G_K$. The same applies to $\mathcal{P}_2 \circ G_K$. These two composite subroutines, $\mathcal{P}_1 \circ G_K$ and $\mathcal{P}_2 \circ G_K$, serve as input parameters for the \texttt{SwapTest} function. Consequently, lines $3\sim 5$ return the results of the Swap Test for the following three pairs of states:

%Here the Pauli input state $\left|\psi_p\right>$ is generated by a pre-prepared unitary operation $G_p$, i.e., $G_p \left|0\right> = \left|\psi_p\right>$. $\mathcal{P}_1 \circ G_p$ means a subroutine which executes $G_p$ and $\mathcal{P}_1$ successively on the input state (here, '$\circ$' represents the sequential execution from right to left). Note that the quantum variables in subroutine \texttt{SwapTest} are all initialized into $\left|0\right>$ (line 3 in Algorithm~\ref{Alg:SwapTest}) and executing $\mathcal{P}_1 \circ G_p$ on $\left|0\right>$ is equivalent to executing $\mathcal{P}_1$ on Pauli state $\left|\psi_p\right>$ (Similarly $\mathcal{P}_2 \circ G_p$). 

\begin{small}
\begin{itemize}
    \item Line 3: $\mathcal{P}_1(\left|\psi_K\right>\left<\psi_K\right|)$, $\mathcal{P}_1(\left|\psi_K\right>\left<\psi_K\right|)$
    \item Line 4: $\mathcal{P}_2(\left|\psi_K\right>\left<\psi_K\right|)$, $\mathcal{P}_2(\left|\psi_K\right>\left<\psi_K\right|)$
    \item Line 5: $\mathcal{P}_1(\left|\psi_K\right>\left<\psi_K\right|)$, $\mathcal{P}_2(\left|\psi_K\right>\left<\psi_K\right|)$ 
\end{itemize}
\end{small}

\noindent
and the results are assigned to $s_1$, $s_2$, and $s_{12}$ denoted in formula (\ref{equ:Et1t2t12}), respectively.
%The Pauli index $p$ is generated randomly in line $2$.
Line $6$ calculates parameter $\tilde{E}$ by formula (\ref{equ:Et1t2t12}). Lines $7\sim 9$ compare $\tilde{E}$ with $\epsilon$, and return FAIL if $\tilde{E}$ is over the limit of $\epsilon$.

The loop is repeated for $k$ rounds (line 2). In each round, three Swap Tests are executed, each requiring $O(ns)$ time and $2n+1$ qubits. So the time complexity of Algorithm~\ref{Alg:EqCheck} is $O(nks)$ and requires $2n+1$ qubits.

\begin{algorithm}
\small
\caption{\texttt{EquivalenceChecking\_original}}
\label{Alg:EqCheck}
\KwIn{($n$, $k$, $s$, $\epsilon$, $\mathcal{P}_1$, $\mathcal{P}_2$) \\ \quad $n$: the number of qubits of the target program; \\ \quad $k$: the number of test points; \\  \quad $s$: the number of rounds of Swap Test; \\ \quad $\epsilon$: tolerance error; \\ \quad $\mathcal{P}_1, \mathcal{P}_2$: two target programs.}
\KwOut{PASS or FAIL.}

\For{i = 1 to k}{
	$K \leftarrow$ Randomly choose a $n$-element integer array, where each integer is in $0,\dots,5$\;
	
	$s_{1} \leftarrow$ \texttt{SwapTest}($n$, $s$, $\mathcal{P}_1 \circ G_K$, $\mathcal{P}_1 \circ G_K$)\;
	
	$s_{2} \leftarrow$ \texttt{SwapTest}($n$, $s$, $\mathcal{P}_2 \circ G_K$, $\mathcal{P}_2 \circ G_K$)\;
	
	$s_{12} \leftarrow$ \texttt{SwapTest}($n$, $s$, $\mathcal{P}_1 \circ G_K$, $\mathcal{P}_2 \circ G_K$)\;
	
	$E\leftarrow (2s_{12} - s_1 - s_2) / s$\;
	
	\If{$|E| > \epsilon$}{\Return{FAIL}\;}
}
\Return{PASS}\;
\end{algorithm}

\subsection{Identity Checking}
\label{subsec:IDcheck}

It seems that identity checking is a special case of equivalence checking - whether the target program \texttt{P} is equivalent to identity program $I$. Fortunately, identity checking has some good properties to avoid repeating running the Swap Test. Thus, it has a faster algorithm than equivalence checking.

Just like equivalence checking discussed in Section~\ref{subsec:EqCheck}, we consider Pauli input state $\left|\psi_K\right>$, which is generated by a unitary transform $G_K$ from all-zero state, where $K$ is the integer array to represent the input Pauli state. Suppose the target program is $\mathcal{P}$. If $\mathcal{P}$ is identity, it keeps any input state, and thus applying the inverse transform $G_K^{-1}$ will recover the state back into the all-zero state. Then, measuring the state will always obtain a result of 0. In other words, after this process, if there is a non-zero result occurs, the target program $\mathcal{P}$ is not identity. Just like the equivalence checking, we test only a small subset (size $k$) of input states for identity checking. During the checking, FAIL is returned whenever one non-zero result occurs; otherwise, PASS is returned.

The algorithm for identity checking is shown in Algorithm~\ref{Alg:IdCheck}.
The algorithm first generates a random integer array $K$ for the input state in line 2 and then initializes the quantum variable \textbf{qs} in line 3. After that, it executes the subroutines $G_p$, $\mathcal{P}$, and $G_p^{-1}$ on \textbf{qs} (line 4) and makes a measurement on \textbf{qs} (line 5). Finally, the algorithm judges whether the measurement result is not 0 (lines $6\sim 8$), and if so, returns FAIL immediately. PASS is returned only when all measurement results are 0 (line 10).

The loop is repeated for $k$ rounds (line 2). So the time complexity of Algorithm~\ref{Alg:IdCheck} is $O(nk)$ and requires $n$ qubits. Compared with equivalence checking (Algorithm~\ref{Alg:EqCheck}), owing to the Swap Test being avoided in identity checking, it is faster and requires fewer qubits. In fact, this algorithm may be faster for checking a non-identity target program due to the immediate return of non-zero measurement results.

\begin{algorithm}
\small
\caption{\texttt{IdentityChecking}}
\label{Alg:IdCheck}
\KwIn{($n$, $k$, $\mathcal{P}$) \\ \quad $n$: the number of qubits of the target program; \\ \quad $k$: the number of test points; \\ \quad $\mathcal{P}$: the target program.}
\KwOut{PASS or FAIL.}

\For{i = 1 to k}{
	$K \leftarrow$ Randomly choose a $n$-element integer array, where each integer is in $0,\dots,5$\;
	\textbf{qs} $\leftarrow \left|0\right>^{\otimes n}$\;
	$G_K$(\textbf{qs});
	$\mathcal{P}$(\textbf{qs});
	$G_K^{-1}$(\textbf{qs})\;
	$r \leftarrow$ Measure(\textbf{qs})\;
	\If{$r \neq 0$}{\Return{FAIL}\;}
}
\Return{PASS}\;
\end{algorithm}

\subsection{Unitarity Checking}
\label{subsec:Ucheck}

\subsubsection{Theoretical foundations}
Note that a unitarity transform exhibits two obvious properties: (1) it preserves the inner product of two input states, and (2) it maintains the purity of an input pure state.  In the context of black-box testing, specifically for unitarity checking, we present a novel theorem that serves as a guiding principle for performing such checking. The formalization of this theorem is provided below.

\begin{theorem}
\label{TheoremUcheck}
Let $\mathcal{H}$ be a $d$-dimensional Hilbert space and \\$\{\left|1\right>,\dots,\left|d\right>\}$ is a group of standard orthogonal basis of $\mathcal{H}$. Let $\left|+_{mn}\right> = \frac{1}{\sqrt{2}}(\left|m\right>+\left|n\right>)$ and $\left|-_{mn}\right> = \frac{1}{\sqrt{2}}(\left|m\right>-\left|n\right>)$, where $m,n=1,\dots,d$ and $m \neq n$. A quantum operation $\mathcal{E}$ is a unitary transform if and only if it satisfies:

\begin{itemize}
\item[(1)] $\forall m \neq n$,  $tr\left[ \mathcal{E}(\left|m\right>\left<m\right|) \mathcal{E}(\left|n\right>\left<n\right|) \right] = 0$;

\item[(2)] for $d-1$ combinations of $(m,n), m \neq n$ which form the edges of a connected graph with vertices $1,\dots,d$, \\
$tr\left[ \mathcal{E}(\left|+_{mn}\right>\left<+_{mn}\right|) \mathcal{E}(\left|-_{mn}\right>\left<-_{mn}\right|) \right] = 0$.
\end{itemize}
\end{theorem}

The proof of Theorem~\ref{TheoremUcheck} is presented in Appendix~\ref{apdsec:ProofUcheck}. It is important to note that $\left|m\right> \bot \left|n\right>$ and $\left|+_{mn}\right> \bot \left|-_{mn}\right>$ hold true when $m\neq n$. Intuitively, Theorem~\ref{TheoremUcheck} establishes that a quantum operation qualifies as a unitary transform if and only if it maps two orthogonal states into two orthogonal states. Essentially, the \textit{preservation of orthogonality} emerges as the fundamental characteristic of unitary transforms. For a more comprehensive exploration of these properties and their impact on the unitarity estimation of quantum channels, Kean et al.~\cite{kean2022unitarity} have conducted an in-depth study.

According to Theorem~\ref{TheoremUcheck}, determining whether a quantum program $\mathcal{P}$ qualifies as a unitary transform necessitates estimating parameters of the form $tr\left( \mathcal{P}(\rho_1)\mathcal{P}(\rho_2) \right)$, where $\rho_1$ and $\rho_2$ denote input states. These input pairs $(\rho_1$, $\rho_2)$ should encompass two categories: (a) pairs of computational basis states $(\left|m\right>, \left|n\right>)$, where $ m \neq n$, and (b) pairs of superposition states $(\left|+_{mn}\right>,\left|-_{mn}\right>)$, where $m\neq n$. Similar to equivalence checking and identity checking, we need to test only a small subset of all possible input pairs, ensuring coverage of both type (a) and type (b). In practical implementation, we employ the Swap Test for evaluating $tr\left( \mathcal{P}(\rho_1)\mathcal{P}(\rho_2) \right)$. Let us assume there are a total of $s$ rounds for the Swap Test on each input, with $s_1$ rounds yielding the result '1'. Based on formula (\ref{equ:SwapTest}), when $tr\left( \mathcal{P}(\rho_1)\mathcal{P}(\rho_2) \right) = 0$, we have $r = 1-\frac{2s_1}{s} = 0$. In practice, we set a tolerance value $\epsilon$. Our testing approach returns a FAIL outcome if there exists a test point where $|r|>\epsilon$, and it returns a PASS only if all test points satisfy $|r|\leq\epsilon$.

\subsubsection{Algorithm}

Based on the preceding discussion, the unitarity checking algorithm for a target quantum program $\mathcal{P}$ involves the following three main steps:

\begin{itemize}
\item[(1)] Randomly generate two types of input state pairs: (a) pairs of computational basis states $(\left|m\right>, \left|n\right>)$, where $ m \neq n$; (b) pairs of superposition states $(\left|+_{mn}\right>,\left|-_{mn}\right>)$, where $m\neq n$. Ensure that each type constitutes approximately half of the total number of test points $k$.

\item[(2)] For each input pair $(\rho_1,\rho_2)$ created in step (1), apply $\mathcal{P}$ to both input states, resulting in the output pair $(\mathcal{P}(\rho_1),\mathcal{P}(\rho_2))$. Employ the Swap Test to determine whether $tr(\mathcal{P}(\rho_1)\mathcal{P}(\rho_2))$ equals 0.

\item[(3)] If any group of output pairs $(\mathcal{P}(\rho_1),\mathcal{P}(\rho_2))$ yields $tr(\mathcal{P}(\rho_1)\mathcal{P}(\rho_2)) \neq 0$, the algorithm reports FAIL; otherwise, if all output pairs yield $tr(\mathcal{P}(\rho_1)\mathcal{P}(\rho_2)) = 0$, it reports PASS.
\end{itemize}

%The condition $i\leq \lceil k/2 \rceil$ means the two sampling cases are about half and half.
The concrete method is presented in Algorithm~\ref{Alg:UnCheck}. The "if" statement in line 2 enacts the classification sampling. Lines 3 to 5 deal with input states of type (b), while lines 7 to 8 manage input states of type (a). In our implementation, for type (b), we select $n$, which is the bitwise negation of $m$, to ensure the superposition covers all qubits (as per the SCAQ criterion in~\cite{long2022process}). The generation processes of computational basis states $\left|m\right>$ and $\left|n\right>$ are denoted as $C_m$ and $C_n$, respectively, while the generation processes of superposition states $\left|+_{mn}\right>$ and $\left|-_{mn}\right>$ are denoted as $S_{m,n}^{+}$ and $S_{m,n}^{-}$. Similar to Algorithm~\ref{Alg:EqCheck}, the notation $\mathcal{P}\circ \mathcal{X}$ signifies the subroutine that executes $\mathcal{P}$ on the input state generated by $\mathcal{X}$. Consequently, lines 5 and 8 implement the required Swap Test on the output pair $(\mathcal{P}(\rho_1),\mathcal{P}(\rho_2))$ to obtain $s_1$ for two different cases, respectively. Line 10 calculates the parameter $tr(\mathcal{P}(\rho_1)\mathcal{P}(\rho_2))$, and line 11 determines whether it equals 0 with a tolerance $\epsilon$. If it does not, it immediately returns FAIL (line 12).

The loop is repeated for $k$ rounds (line 2). In each round, three Swap Tests are executed, each requiring $O(ns)$ time and $2n+1$ qubits. So the time complexity of Algorithm~\ref{Alg:UnCheck} is $O(nks)$ and requires $2n+1$ qubits.

\begin{algorithm}
\small
\caption{\texttt{UnitarityChecking\_original}}
\label{Alg:UnCheck}
\KwIn{($n$, $k$, $s$, $\epsilon$, $\mathcal{P}$) \\ \quad $n$: the number of qubits of the target program; \\ \quad $k$: the number of test points; \\ \quad $s$: the number of rounds of Swap Test; \\ \quad $\epsilon$: tolerance error; \\ \quad $\mathcal{P}$: the target program.}
\KwOut{PASS or FAIL.}

\For{i = 1 to k}{
    \eIf{$i \leq \lceil k/2 \rceil$}
    {
    	$m \leftarrow$ Randomly choose a number from 0 to $2^n-1$\;
    	$n \leftarrow$ Bitwise negation of $m$\;
    	$s_1 \leftarrow$ SwapTest($n$, $s$, $\mathcal{P}\circ S_{m,n}^{+}$, $\mathcal{P}\circ S_{m,n}^{-}$)\;
    }
    {
    	$m, n \leftarrow$ Randomly choose two different numbers from 0 to $2^n-1$\;
    	$s_1 \leftarrow$ SwapTest($n$, $s$, $\mathcal{P}\circ C_m$, $\mathcal{P}\circ C_n$)\;
    }
    $r \leftarrow 1 - (2s_1 / t)$\;
	\If{$|r| > \epsilon$}{\Return{FAIL}\;}
}
\Return{PASS}\;
\end{algorithm}

\subsection{Parameter Selection}
\label{subsec:parameter}

In Algorithms~\ref{Alg:EqCheck}, \ref{Alg:IdCheck}, and \ref{Alg:UnCheck}, several parameters need to be selected by users - number of test points $k$, number of rounds of Swap Test $s$, and the tolerance error $\epsilon$ (the latter two parameters are required in equivalence checking and unitarity checking because they are based on Swap Test). In this section, we discuss how these parameters influence the effectiveness of our testing methods and how to select these parameters.

Due to quantum programs' nondeterministic nature, testing algorithms may output wrong results. Consider that the target program has two statuses: \textit{Correct} or \textit{Wrong}, and the output also has two statuses: \textit{PASS} or \textit{FAIL}. So, there are two types of wrong results:

\begin{itemize}
\item \textbf{Error Type I:} A wrong program passes;
\item \textbf{Error Type II:} A correct program fails.
\end{itemize}

The general principle of parameter selection is to balance the accuracy and running time. To increase the accuracy, we should try to control the probabilities of both two types of errors.

In identity checking, the unique parameter is the number of test points $k$. Because an identity program always returns a result 0 in identity checking, a type II error will not occur. For a non-identity program, the average probability of returning result 0, denoted as $p$, satisfies $0<p<1$. Taking $k$ test points, the probability of type I error is $p^k$. It means that with the increase of $k$, the probability of type I error decreases exponentially. Therefore, selecting a moderately sized value for $k$ is sufficient to control the probability of errors.

In the equivalence and unitarity checking, there are three parameters: number of test points $k$, number of rounds of Swap Test $s$, and tolerance $\epsilon$. They are based on statistics, and both type I and type II errors can possibly occur due to the estimation errors. Intuitively, larger $k$ (testing more points) and smaller $\epsilon$ (more strict judgment condition for the correctness) are helpful to reduce type I error. However, wrong programs are diverse, and we cannot know any information about errors before testing. So, the selection of $k$ and $\epsilon$ is more empirical. In Section~\ref{sec:Experments}, we will find an appropriate selection of $k$ and $\epsilon$ according to some benchmark programs.

If we choose a smaller $\epsilon$, i.e., the judgment condition for the correctness is more strict, we need to improve the accuracy of the Swap Test to avoid type II error, i.e., select a larger $s$. Fortunately, owing to the behavior of the correct program is unique, given $k$ and $\epsilon$, the parameter $s$ can be calculated before testing. We need an extra parameter $\alpha_2$, the required maximum probability of type II error. Then, we have the following proposition:

%In fact, $k$ and $\epsilon$ determine the ability to discover error programs. Large $k$ and small $\epsilon$ can discover more potential errors but are more likely to misjudge correct programs. Fortunately, \red{we can choose a property $s$} to control the probability of type 2 error not to exceed $\alpha_2$. 

\begin{proposition}
\label{prop:SelectT}
In equivalence checking or unitarity checking, suppose we have selected the number of test points $k$ and tolerance $\epsilon$. Given the allowed probability of type II error $\alpha_2$. If the number of rounds of Swap Test $s$ satisfies:

\begin{itemize}
\item[(1)] $s \geq \frac{8}{\epsilon^2}\ln \frac{2}{1-(1-\alpha_2)^{\frac{1}{k}}}$ for equivalence checking, or

\item[(2)] $s \geq \frac{2}{\epsilon^2 \ln{2}}\ln \frac{1}{1-(1-\alpha_2)^{\frac{1}{k}}}$ for unitarity checking,
\end{itemize}

\noindent then the probability of type II error will not exceed $\alpha_2$.
\end{proposition}

As a practical example, Table~\ref{tab:T} gives the selection of $s$ on several groups of $k$ and $\epsilon$ under $\alpha_2 = 0.1$. The lower bound of $s$ seems complicated, but we can prove:

\begin{equation}
\label{equ:Tbound}
\frac{k}{-\ln(1-\alpha_2)} \leq \frac{1}{1-(1-\alpha_2)^\frac{1}{k}} \leq  \frac{k}{\alpha_2}
\end{equation}

\noindent where $k\geq 1$ and $0<\alpha_2<1$. It deduces:

\begin{equation}
\label{equ:ThetaK}
\frac{1}{1-(1-\alpha_2)^\frac{1}{k}} = \Theta(k)
\end{equation}

\noindent The proofs of Proposition~\ref{prop:SelectT} and formula (\ref{equ:Tbound}) are shown in Appendix~\ref{apdsec:ProofT}. Proposition~\ref{prop:SelectT} and formula (\ref{equ:ThetaK}) deduce the following corollary about the asymptotic relations of $s$ and the overall time complexity of Algorithms~\ref{Alg:EqCheck} and~\ref{Alg:UnCheck}:

\begin{corollary}
\label{cor:t}
In both equivalence checking and unitarity checking, for a fixed $\alpha_2$, if we choose $s$ according to the lower bound in Proposition~\ref{prop:SelectT}, then $s = \Theta\left(\frac{\log k}{\epsilon^2}\right)$. Further more, when adopting this choice, the time complexities of Algorithm~\ref{Alg:EqCheck} and~\ref{Alg:UnCheck} are both $O\left(\frac{nk \log k}{\epsilon^2}\right)$.
\end{corollary}

\begin{table}
\centering
\caption{The choice of $s$ with different $k$ and $\epsilon$ under $\alpha_2 = 0.1$}
\label{tab:T}
\resizebox{\linewidth}{!}{
\begin{tabular}{c|cccc|cccc}
	\toprule
	& \multicolumn{4}{c|}{Equivalence Check} & \multicolumn{4}{c}{Unitarity Check}\\
	\hline
	\diagbox{$k$}{$\epsilon$} & 0.05 & 0.10 & 0.15 & 0.20 & 0.05 & 0.10 & 0.15 & 0.20\\
	\midrule
	1 & 9587 & 2397 & 1066 & 600 &-&-&-&-\\
	2 & 11722 & 2931 & 1303 & 733 & 3428 & 857 & 381 & 215\\
	3 & 12991 & 3248 & 1444 & 812 & 3886 & 972 & 432 & 243\\
	4 & 13898 & 3475 & 1545 & 869 & 4213 & 1054 & 469 & 264\\
	6 & 15181 & 3796 & 1687 & 949 & 4676 & 1169 & 520 & 293\\
	10 & 16804 & 4202 & 1868 & 1051 & 5261 & 1316 & 585 & 329\\
	\bottomrule
\end{tabular}
}
\end{table}

In Section~\ref{sec:Experments}, we will also provide experimental research on the parameters $k$, $s$, and $\epsilon$.

\section{Optimization for Equivalence and Unitarity Checking}
\label{sec:optimization}

\subsection{Optimization Ideas}
\label{subsec:optidea}
The original algorithms for equivalence checking and unitarity checking (Algorithms~\ref{Alg:EqCheck} and~\ref{Alg:UnCheck}) involve repeated running the Swap Test many times to obtain the indicating parameters of equivalence and unitarity. In contrast, identity checking (Algorithm~\ref{Alg:IdCheck}) immediately returns FAIL upon encountering a non-zero result. Thus, it enhances the efficiency of discovering bugs. This observation prompts us to seek similar properties for equivalence and unitarity checking that allow us to confirm the existence of errors immediately upon a certain result rather than waiting for all results.

As discussed in Section~\ref{subsec:swap_test}, formula (\ref{equ:SwapTest}) about Swap Test gives $tr(\rho_1\rho_2) = 1 - 2p_1$. So $tr(\rho_1\rho_2)=1$ if and only if $p_1$, the probability of obtaining result '1', equals 0. In other words, as long as a result '1' occurs, we can immediately conclude that $tr(\rho_1\rho_2)\neq 1$. Thus, we introduce Algorithm~\ref{Alg:TrABequ1}, to fast determine whether $tr(\rho_1\rho_2) = 1$. The main body of Algorithm~\ref{Alg:TrABequ1} is the same as Algorithm~\ref{Alg:SwapTest} about Swap Test. The difference is that Algorithm~\ref{Alg:TrABequ1} returns FALSE immediately when result '1' occurs (lines $10\sim 12$), while Algorithm~\ref{Alg:SwapTest} accumulates the number of results '1'.

%identity checking circumvents the need for the Swap Test by leveraging a specific property of output states. It returns FAIL immediately as long as there exists an output dissatisfying the property.
% Naturally, if we can find similar properties for equivalenc and unitarity checking, the efficiency of discovering bugs for them can be also enhanced.
%Fortunately, Swap Test possesses such a property, particularly when dealing with pure states.
%In this section, we aim to identify properties of equivalence checking and unitarity checking that can optimize their runtime efficiency.
%In equivalence and unitarity checking, Swap Test is utilized to estimate the parameter $tr(\rho\sigma)$. As discussed in Section~\ref{subsec:swap_test}, we have $tr(\rho\sigma) = 1 - 2p_1$, where $p_1$ represents the probability of obtaining the '1' result. Notably, $tr(\rho\sigma)=1$ if and only if $p_1=0$, indicating that the occurrence of the '1' result allows us to conclude that $tr(\rho\sigma) \neq 1$. Consequently, we introduce Algorithm~\ref{Alg:TrABequ1} to determine whether $tr(\rho_1\rho_2)$ equals 1 for a given pair of input states, $\rho_1$ and $\rho_2$. The main body of this algorithm is identical to Algorithm~\ref{Alg:SwapTest}, except for its immediate return of FALSE upon measuring the result '1'.

Note that both equivalence and unitarity checking are based on the Swap Test. To optimize them using Algorithm~\ref{Alg:TrABequ1}, it is necessary to find their properties related to whether $tr(\rho_1\rho_2)$ equals 1. Obviously, there are two related properties:

\begin{itemize}
\item[(1)] The \textit{purity checking} for a state:

By setting $\rho_1=\rho_2=\rho$, Algorithm~\ref{Alg:TrABequ1} is able to determine whether $tr(\rho^2)=1$, thereby assessing whether $\rho$ represents a pure state.

\item[(2)] The \textit{equality checking} for two pure states:

Let $\rho_1 = \left|\alpha\right>\left<\alpha\right|$ and $\rho_2 = \left|\beta\right>\left<\beta\right|$ are two pure states, then $tr(\rho_1\rho_2) = \left|\left<\alpha|\beta\right>\right|^2$, and it equals 1 if and only if $\left|\alpha\right>=\left|\beta\right>.$
\end{itemize}

In the following part of this section, we will identify and examine the properties of equivalence and unitarity checking and delve into the strategies for optimizing these properties using Algorithm~\ref{Alg:TrABequ1}.

\begin{algorithm}
\small
\caption{\texttt{is\_TrAB\_equals\_1} }
\label{Alg:TrABequ1}
\KwIn{($n$, $s$, $\mathcal{P}_2$, $\mathcal{P}_2$) \\ \quad $n$: the number of qubits of one target state; \\ \quad $s$: the number of rounds; \\ \quad $\mathcal{P}_1, \mathcal{P}_2$: subroutines to generate target states $\rho_1$ and $\rho_2$.}
\KwOut{TRUE or FALSE about whether $tr(\rho_1\rho_2)=1$.}

\For{i = 1 to $s$}{
	\textbf{qanc} $\leftarrow \left|0\right>$; \textbf{qs1} $\leftarrow \left|0\right>^{\otimes n}$; \textbf{qs2} $\leftarrow \left|0\right>^{\otimes n}$\;
	$H(\textbf{qanc})$\;
	$\mathcal{P}_1(\textbf{qs1})$; $\mathcal{P}_2(\textbf{qs2})$\;
	\For{j = 0 to $n-1$}
	{
		Controlled SWAP(\textbf{qanc}, (\textbf{qs1}[j], \textbf{qs2}[j]))\;
	}
	$H(\textbf{qanc})$\;
	$m \leftarrow$ Measure(\textbf{qanc})\;
	\If{$m = 1$}{\Return{FALSE};}
}
\Return{TRUE}\;
\end{algorithm}

\subsection{Optimized Equivalence Checking}
\label{subsec:OptEqCheck}

As discussed in Section~\ref{subsec:EqCheck}, equivalence checking relies on the evaluation of the following parameter

$$ E(\rho_1,\rho_2) = \left|\frac{tr(\rho_1^2)+tr(\rho_2^2)}{2} - tr(\rho_1\rho_2)\right|$$

\noindent and determines whether it equals 0. The values of $tr(\rho_1^2)$, $tr(\rho_2^2)$, and $tr(\rho_1\rho_2)$ can be estimated using the Swap Test. When both $\rho_1$ and $\rho_2$ are pure states (i.e., $tr(\rho_1^2) = tr(\rho_2^2) = 1$), we have $E(\rho_1,\rho_2) = 0$ if and only if $tr(\rho_1\rho_2) = 1$. As discussed in Section~\ref{subsec:optidea}, purity checking for a single state and determining whether $tr(\rho_1\rho_2) = 1$ for two pure states can be accomplished fast using Algorithm~\ref{Alg:TrABequ1}. The idea of the optimization approach involves first checking whether the input states are both pure states. If so, Algorithm~\ref{Alg:TrABequ1} is directly employed to assess whether $tr(\rho_1\rho_2) = 1$, bypassing the need for the general Swap Test. Consider that numerous practical quantum programs execute unitary transform, thereby preserving the purity of output states when given pure input states. Consequently, this optimization strategy may be effective for such programs.

The optimized algorithm is presented in Algorithm~\ref{Alg:OptEqCheck}. Lines $1\sim 2$ and $14\sim 23$ are the same as Algorithm~\ref{Alg:EqCheck}, while lines $3\sim 13$ are newly introduced for optimization. Similar to the discussion in Section~\ref{subsec:EqCheck}, for $i=1,2$, $\mathcal{P}_i \circ G_K$ generates the state of the output of $\mathcal{P}_i$ under the input generated by $G_K$. So lines $3\sim 4$ invoke Algorithm~\ref{Alg:TrABequ1} to verify the purity of the output states from the two target programs, respectively. The additional input parameter $t$ represents the number of rounds in Algorithm~\ref{Alg:TrABequ1}, while the number of rounds in the Swap Test is denoted as $s$. If their purities differ (one state is pure, but the other is not, line 5), it can be concluded that the output states are distinct, and consequently, the target programs are not equivalent (line 6). If both states are pure, Algorithm~\ref{Alg:TrABequ1} is employed to assess their equivalence (lines $9\sim 12$), providing a fast determination. Otherwise, the original equivalence checking algorithm must be executed (lines $14\sim 23$).

\begin{algorithm}
\small
\caption{\texttt{EquivalenceChecking\_optimized}}
\label{Alg:OptEqCheck}
\KwIn{($n$, $k$, $t$, $s$, $\epsilon$, $\mathcal{P}_1$, $\mathcal{P}_2$) \\ \quad $n$: the number of qubits of the target program; \\ \quad $k$: the number of test points $k$; \\  \quad $t$: the number of rounds of trace-1-checking subroutine; \\ \quad $s$: the number of rounds of Swap Test; \\ \quad $\epsilon$: tolerance error; \\ \quad $\mathcal{P}_1, \mathcal{P}_2$: two target programs.}
\KwOut{PASS or FAIL.}

\For{i = 1 to k}{
	$K \leftarrow$ Randomly choose a $n$-element integer array, where each integer is in $0,\dots,5$\;
	$isPure1 \leftarrow$ \texttt{is\_TrAB\_equals\_1}($n$, $t$, $\mathcal{P}_1 \circ G_K$, $\mathcal{P}_1 \circ G_K$)\;
	$isPure2 \leftarrow$ \texttt{is\_TrAB\_equals\_1}($n$, $t$, $\mathcal{P}_2 \circ G_K$, $\mathcal{P}_2 \circ G_K$)\;
	
	\If{$isPure1 \neq isPure2$}{ \Return{FAIL}\; }
	
	\eIf{$isPure1 =$ TRUE}{
		$TrEq \leftarrow$ \texttt{is\_TrAB\_equals\_1}($n$, $t$, $\mathcal{P}_1 \circ G_K$, $\mathcal{P}_2 \circ G_K$)\;
		\If{$TrEq =$ FALSE}{ \Return{FAIL}\; }
	}{ %else
		$s_{1} \leftarrow$ \texttt{SwapTest}($n$, $s$, $\mathcal{P}_1 \circ G_K$, $\mathcal{P}_1 \circ G_K$)\;
		$s_{2} \leftarrow$ \texttt{SwapTest}($n$, $s$, $\mathcal{P}_2 \circ G_K$, $\mathcal{P}_2 \circ G_K$)\;
		$s_{12} \leftarrow$ \texttt{SwapTest}($n$, $s$, $\mathcal{P}_1 \circ G_K$, $\mathcal{P}_2 \circ G_K$)\;
		$E\leftarrow (2s_{12} - s_1 - s_2) / s$\;
		\If{$|E| > \epsilon$}{\Return{FAIL}\;}
	}
}
\Return{PASS}\;
\end{algorithm}

\subsection{Optimized Unitarity Checking}
\label{subsec:OptUnCheck}

The original unitarity checking algorithm is based on checking the orthogonal preservation condition, requiring the evaluation of $tr(\rho_1\rho_2)=0$ for two states $\rho_1$ and $\rho_2$. However, $tr(\rho_1\rho_2)=0$ if and only if the probability $p_1$ of obtaining the '1' result equals $1/2$. It implies that we need to check if the '1' result occurs approximately half of the time, unable to return immediately.

The idea of optimization is to focus on another crucial property of unitary transforms: \textit{purity preservation}. A unitary transform must map a pure state into another pure state. If a program maps the input pure state into a mixed state, it can be immediately concluded that the program is not a unitary transform. By employing Algorithm~\ref{Alg:TrABequ1}, we can perform purity checking. However, purity preservation alone is insufficient to guarantee unitarity. For instance, the \texttt{Reset} procedure maps any input state to the \textit{all-zero state} $\left|0 \dots 0\right>$. While this transformation maps a pure state to another pure state (since the all-zero state is pure), it is not a unitary transform. Thus, the original checking approach that relies on \textit{orthogonality preservation} remains essential.

The optimized algorithm, presented in Algorithm~\ref{Alg:OptUnCheck}, enhances the original approach. It includes a purity checking step for the output state of the target program $\mathcal{P}$ when subjected to Pauli inputs (generated by $G_K$ with random $K$) using Algorithm~\ref{Alg:TrABequ1} (lines $2\sim 3$). If the purity checking returns FALSE (line 4), the algorithm immediately fails (line 5). Similar to the equivalence checking, we introduce a new parameter $t$ to represent the number of rounds in Algorithm~\ref{Alg:TrABequ1}, and $s$ denotes the number of rounds in the Swap Test. Upon successful purity checking, the original unitarity checking should be executed (lines $8\sim 9$).

\begin{algorithm}
\small
\caption{\texttt{UnitarityChecking\_optimized}}
\label{Alg:OptUnCheck}
\KwIn{($n$, $k$, $t$, $s$, $\epsilon$, $\mathcal{P}$) \\ \quad $n$: the number of qubits of the target program; \\ \quad $k$: the number of test points; \\ \quad $t$: the number of rounds of trace-1-checking subroutine; \\ \quad $s$: the number of rounds of Swap Test; \\ \quad $\epsilon$: tolerance error; \\ \quad $\mathcal{P}$: the target program.}
\KwOut{PASS or FAIL.}

\For{i = 1 to k}{
	$K \leftarrow$ Randomly choose a $n$-element integer array, where each integer is in $0,\dots,5$\;
	$isTr1 \leftarrow$\texttt{is\_TrAB\_equals\_1}($n$, $t$, $\mathcal{P} \circ G_K$, $\mathcal{P} \circ G_K$)\;
	\If{$isTr1 = $ FALSE}{\Return{FAIL}\;}
}
$r\leftarrow$ \texttt{UnitarityChecking\_original}($n$, $k$, $s$, $\epsilon$, $\mathcal{P}$)\;
\Return{r}
\end{algorithm}

\subsection{Selection of $t$}
\label{subsec:ttrace}

In Algorithms~\ref{Alg:OptEqCheck} and~\ref{Alg:OptUnCheck}, the parameters $k$, $s$, and $\epsilon$ are selected using the same strategy as the original algorithms, which has been discussed in Section~\ref{subsec:parameter}. However, the optimized algorithms introduce an additional parameter $t$, which should have a balanced choice. If $t$ is too small, there is a risk of misjudging a mixed state as a pure state since both results, '0' and '1', are possible for mixed states. This can lead to subsequent misjudgments. On the other hand, if $t$ is too large, it may reduce the efficiency. In Section~\ref{sec:Experments}, we will conduct experimental research on the parameter $t$ to determine an appropriate value.

\section{Experimental Evaluation}
\label{sec:Experments}

In this section, we present the experimental evaluation of our algorithms about equivalence (EQ) checking, identity (ID) checking, and unitarity (UN) checking.

\subsection{Experimental Design}
\label{subsec:expdesign}

Each of our checking algorithms yields a PASS or FAIL result for the given input programs or pairs. For a comprehensive evaluation, it is essential to prepare two categories of targets: \textit{expected-pass} and \textit{expected-fail}. These are used to assess the algorithms' performance on programs that either meet or do not meet the specified relations. An effective algorithm should return PASS for \textit{expected-pass} targets and FAIL for \textit{expected-fail} targets. Therefore, our benchmark programs include both types of targets for each checking task. Existing benchmarks, such as those in~\cite{zhao2021bugs4q}, are designed for general testing and not for the specific tasks discussed in our paper, making them unsuitable for evaluating our relation-focused testing task. To remedy this, we are creating a specialized program benchmark, the details of which will be outlined in Section~\ref{subsec:benchmark}.

As outlined in Sections~\ref{sec:methods} and~\ref{sec:optimization}, our checking algorithms are influenced by several adjustable parameters. A key aspect of our evaluation involves investigating how these parameters affect the algorithms' accuracy. Moreover, we aim to compare the efficacy of the original and optimized versions of our EQ and UN checking algorithms. Additionally, a comprehensive assessment of these algorithms' overall performance, including an exploration of cases where they underperform, is a crucial part of our evaluation. To guide this process, we will focus on the following research questions (RQs):

\begin{itemize}
\item \textbf{RQ1:} How do the parameters $k$ and $\epsilon$ affect the bug detection capability of the original algorithms?

\item \textbf{RQ2:} Is the selection of the $s$ value, as presented in Proposition~\ref{prop:SelectT}, effective for practical testing tasks using the original algorithms?

\item \textbf{RQ3:} What is the performance of the optimized algorithms in terms of efficiency and bug detection accuracy?

\item \textbf{RQ4:} How well do our checking methods perform on the benchmark programs?

\item \textbf{RQ5:} What are the characteristics of the cases on which our algorithms do not perform so well?
\end{itemize}

Our experiments do not include comparisons with existing methods due to two key reasons: Firstly, our approach to EQ and ID checking differs from current methods, which primarily concentrate on quantum information theory~\cite{janzing2005non,Flammia2011FewPauli,montanaro2013survey}  or white-box checking~\cite{viamontes2007checking,yamashita2010fast,burgholzer2020advanced,hong2021approximate,hong2022eqchk_dynamic,wang2022eqchk_seq}, as opposed to our focus on black-box testing. Secondly, UN checking represents a novel problem area for which no existing methods have been established.

To facilitate our experiments, we have developed a dedicated tool for implementing equivalence, identity, and unitarity checking of Q\# programs. This tool enables us to assess the performance and effectiveness of our algorithms effectively. We conducted our experiments using Q\# language (SDK version 0.26.233415) and its simulator. 

The experiments were conducted on a personal computer equipped with an Intel Core i7-1280P CPU and 16 GB RAM. To ensure accuracy and minimize the measurement errors resulting from CPU frequency reduction due to cooling limitations, we executed the testing tasks sequentially, one after another.

\subsection{Benchmark Programs}
\label{subsec:benchmark}

\subsubsection{Original programs}
\label{subsubsec:original}

We select several typical quantum programs as original programs to construct the benchmark, as Table~\ref{table:original-programs} shows.  The "\#Qubits" column indicates the number of used qubits, while the "\#QOps" column provides an approximate count of the single- or two-qubit gates for each program\footnote{Q\# programs are allowed to include loops. The count of gates is determined by (1) fully expanding all loops and (2) decomposing all multi-qubit gates or subroutines into basic single- or two-qubit gates. Due to the possibility of different decomposition methods, \#QOps just represents an approximate count for each program.}, which indicates the program's scale. The "UN" column indicates whether the program implements a unitary transform. Similarly, the "ID" column indicates whether the program is identity.

We have considered the properties of our tasks about EQ, ID, and UN checking in the selection of original programs. First, we select two sets of equivalent circuit pairs as shown in Figure~\ref{fig:TwoCir}: one set is unitary, while the other is not. \texttt{Cir1A} and \texttt{Cir1B} (Figure~\ref{fig:TwoCir}(a)) are from~\cite{jessica2021quanto}, and \texttt{Cir2A} and \texttt{Cir2B} (Fiture~\ref{fig:TwoCir}(b)) are from~\cite{juan2011eqivalent}. Then, we select two typical identity programs: \texttt{Empty} and \texttt{TeleportABA}. Considering that \texttt{P}$\circ$\texttt{invP} is a common identity relation, we select two unitary programs \texttt{QFT}, \texttt{QPE} and their inverse programs \texttt{invQFT}, \texttt{invQPE}. We also select a complicated program \texttt{CRot} and a common non-unitary program \texttt{Reset}.

\begin{table*}
\centering
\caption{Original programs to construct the benchmark}
\label{table:original-programs}
\footnotesize

\begin{tabular}{c|p{10cm}|c|c|c|c}
	\toprule
	\textbf{Program} & \makecell[c]{\textbf{Description}} & \textbf{\#Qubits} & \textbf{\#QOps} & \textbf{UN} & \textbf{ID} \\
	\bottomrule
	\texttt{Cir1A} & The implementation the circuit in Figure~\ref{fig:TwoCir}(a) left. & 2 & 6 & \checkmark &\\
	\hline
	\texttt{Cir1B} & The implementation of the circuit in Figure~\ref{fig:TwoCir}(a) right. It is equivalent to \texttt{Cir1A}. & 2 & 1 & \checkmark &\\
	\hline
	\texttt{Cir2A} & The implementation of the circuit in Figure~\ref{fig:TwoCir}(b) left. & 3 & 4 &&\\
	\hline
	\texttt{Cir2B} & The implementation the circuit in Figure~\ref{fig:TwoCir}(b) right. It is equivalent to \texttt{Cir2A}. & 3 & 4 &&\\
	\hline
	\texttt{Empty} & A program with no statement; thus, it is identity. & 6 & 0 & \checkmark & \checkmark \\
	\hline
	\multirow{3}{*}{\texttt{TeleportABA}} & A program using quantum teleportation to teleport quantum state from input variable \verb+A+ to intermediate variable \verb+B+, and then teleport \verb+B+ back to \verb+A+. So it is the identity on input variable \verb+A+. & \multirow{3}{*}{3} & \multirow{3}{*}{60} & \multirow{3}{*}{\checkmark} & \multirow{3}{*}{\checkmark} \\
	\hline
	\texttt{QFT} & An implementation of the Quantum Fourier Transform (QFT) algorithm. & 5 & 17 & \checkmark & \\
	\hline
	\texttt{invQFT} &  The inverse program of \texttt{QFT}. & 5 & 17 & \checkmark & \\
	\hline
	\multirow{3}{*}{\texttt{QPE}} & An implementation of the Quantum Phase Estimation algorithm. It has an input parameter of estimated unitary transform $U$. In our benchmark, we fix this parameter as $U=SH$. & \multirow{3}{*}{5} & \multirow{3}{*}{22} & \multirow{3}{*}{\checkmark} &\\
	\hline
	\texttt{invQPE} & The inverse program of \texttt{QPE}. & 5 & 22 & \checkmark &\\
	\hline
	\multirow{4}{*}{\texttt{CRot}} & The key subroutine in \textit{Harrow-Hassidim-Lloyd} (HHL) algorithm~\cite{harrow2009quantum}. It implements the unitary transform: $\left|\lambda\right>\left|0\right> \rightarrow \left|\lambda\right>(\sqrt{1-\frac{c^2}{\lambda^2}}\left|0\right> + \frac{c}{\lambda}\left|1\right>)$, where $c$ represents a parameter. In our benchmark, we fix $c=1$. & \multirow{4}{*}{5} & \multirow{4}{*}{680} & \multirow{4}{*}{\checkmark} & \\
	\hline
	\multirow{2}{*}{\texttt{Reset}} & A subroutine that resets a qubit register to the all-zero state $\left|0\dots 0\right>$. Note that it is not a unitarity transform. & \multirow{2}{*}{6} & \multirow{2}{*}{12} & & \\
	\bottomrule
\end{tabular}
\end{table*}

\subsubsection{Construct the benchmark}
\label{subsubsec:benchmark}

We begin by constructing the benchmark based on the original programs, as discussed in Section~\ref{subsec:expdesign}. Our benchmark must include both expected-pass and expected-fail targets for EQ checking. Expected-pass pairs consist of an original program and a carefully selected equivalent program, while expected-fail pairs consist of an original program and a mutated program obtained by applying a mutant operator into the original program. In our evaluation, we consider two types of mutant operators: gate mutant (GM) operators, which involve the addition, removal, or modification of gates, and measurement mutant (MM) operators, which encompass the addition or removal of measurement operations~\cite{fortunato2022mutation,long2022process}. The concrete benchmark is outlined in Table~\ref{tab:benchmark}. Programs labeled with the prefix "\textit{error}" are mutant versions of the original programs generated using either GM or MM. "$5\times$(error \texttt{Cir1A}, \texttt{Cir1B})" signifies that we have prepared five pairs, where each pair consists of one mutant \texttt{Cir1A} and the original \texttt{Cir1B}, and different pair features a different mutant operator.

For ID checking evaluation, the expected-pass programs are constructed based on known identity relations, such as \texttt{P$\circ$invP}, to assess the effectiveness of testing inverse programs (see Scenario~\ref{scn:inverse} in Section~\ref{sec:questions}). Notably, both No.4 and No.11 pertain to the relation \texttt{QFT$\circ$invQFT}. No.4 utilizes EQ checking with \texttt{Empty}, while No.11 employs ID checking. These cases are used to compare ID checking and EQ checking with identity. The expected-fail programs are generated by applying GM or MM mutant operator to each expected-pass program.

For UN checking evaluation, the expected-pass programs comprise those listed in Table~\ref{table:original-programs} that do not contain any measurement operations. Since GM mutations do not affect unitarity while MM mutations do, some expected-pass programs are constructed by applying GM mutant operators, while all expected-fail programs are created by applying MM mutant operators.

As illustrated in Table~\ref{tab:benchmark}, our benchmark comprises 63 Q\# programs or program pairs used to evaluate the effectiveness of our methods.
%We provide additional explanations about some of the programs as follows.

\begin{figure}
%\textbf{}
\centering
\subfigure[\texttt{Cir1}]{\includegraphics[scale=0.5]{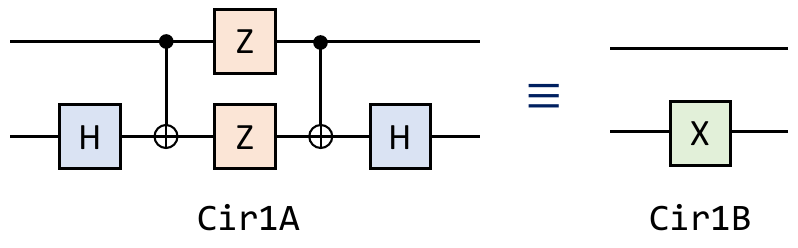}}
\quad
\subfigure[\texttt{Cir2}]{\includegraphics[scale=0.5]{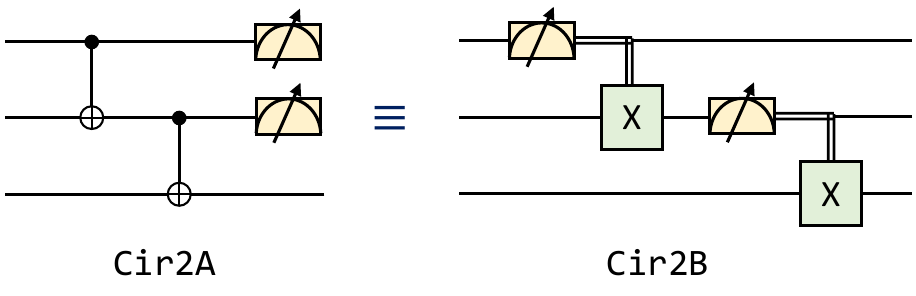}}
\caption{Two sets of equivalent circuits}
\label{fig:TwoCir}
\end{figure}

\begin{table}
\centering
\caption{Benchmark programs}
\label{tab:benchmark}
\footnotesize
\begin{tabular}{c|c|cc|c}
	\toprule
	\textbf{Task} & \textbf{Expected} & \textbf{No.} & \textbf{Programs / Pairs} & \textbf{\#} \\
	\midrule
	
	\multirow{10}{*}{EQ} & \multirow{4}{*}{PASS} & 1 & (\texttt{Cir1A}, \texttt{Cir1B}) & \multirow{4}{*}{4} \\
	&& 2 & (\texttt{Cir2A}, \texttt{Cir2B}) & \\
	&& 3 & (\texttt{QFT}, \texttt{QFT}) & \\
	&& 4 & (\texttt{QFT$\circ$invQFT}, \texttt{Empty}) & \\
	
	\cmidrule{2-5}
	
	& \multirow{5}{*}{FAIL} & 5 & $5\times$(error \texttt{Cir1A}, \texttt{Cir1B}) & \multirow{5}{*}{20}  \\
	&& 6 & $2\times$(error \texttt{Cir2A}, \texttt{Cir2B}) & \\
	&& 7 & $3\times$(\texttt{Cir2A}, error \texttt{Cir2B}) & \\
	&& 8 & $5\times$(error \texttt{QFT}, \texttt{QFT}) & \\
	&& 9 & $5\times$(\texttt{QFT}$\circ$error \texttt{invQFT}, \texttt{Empty}) & \\
	
	\midrule
	
	\multirow{9}{*}{ID} & \multirow{4}{*}{PASS} & 10 & \texttt{Empty} & \multirow{4}{*}{4} \\
	&& 11 & \texttt{QFT$\circ$invQFT} & \\
	&& 12 & \texttt{QPE$\circ$invQPE} & \\
	&& 13 & \texttt{TeleportABA} & \\
	
	\cmidrule{2-5}
	
	& \multirow{4}{*}{FAIL} & 14 & $2\times$ error \texttt{Empty} & \multirow{4}{*}{17} \\
	&& 15 & $5\times$ \texttt{QFT}$\circ$error 
	\texttt{invQFT} & \\
	&& 16 & $5\times$ \texttt{QPE}$\circ$error \texttt{invQPE} & \\
	&& 17 & $5\times$ error \texttt{TeleportABA} & \\
	
	\midrule
	
	\multirow{11}{*}{UN} & \multirow{6}{*}{PASS} & 18 & \texttt{Cir1A} & \multirow{6}{*}{8} \\
	&& 19 & \texttt{Cir1B} & \\
	&& 20 & \texttt{Empty} & \\
	&& 21 & \texttt{QFT} & \\
	&& 22 & \texttt{CRot} & \\
	&& 23 & $3\times$ GMs of \texttt{QFT} & \\
	
	\cmidrule{2-5}
	
	& \multirow{5}{*}{FAIL} & 24 & \texttt{Cir2A} & \multirow{5}{*}{10} \\
	&& 25 & \texttt{Cir2B} & \\
	&& 26 & \texttt{Reset} & \\
	&& 27 & $5\times$ MMs of \texttt{QFT} & \\
	&& 28 & $2\times$ MMs of \texttt{CRot} & \\
	
	\bottomrule
\end{tabular}
\end{table}

%\subsection{Research Questions}
%\label{subsec:RQs}

\subsection{Experiment Configurations}
\label{subsec:ExConfig}

We describe our experimental configurations for each research question as follows\footnote{ The source code of our algorithms and evaluation are available at \href{https://github.com/MgcosA/EvaluationCodeOfQuantumRelationChecking/}{https://github.com/MgcosA/EvaluationCodeOfQuantumRelationChecking} }.

\vspace*{1mm}
\noindent
$\bullet$ \textbf{RQ1:} 
We ran our algorithms on each expected-fail program, adjusting the values of $k$ and $\epsilon$. The selection of these ranges and values was determined by proper preliminary experiments.
We examine $\epsilon$ values of 0.05, 0.10, 0.15, and 0.20 for EQ and UN. For EQ, we test $k$ values of 1, 2, 3, 4, 6, and 10. Since the UN algorithm requires two types of inputs, the minimum $k$ value we consider is 2. Thus, we explore $k$ values of 2, 3, 4, 6, and 10 for UN. The corresponding $s$ values are determined as shown in Table~\ref{tab:T}. For ID, the algorithm does not involve the Swap Test, allowing for larger $k$ values. Therefore, we select $k$ values of 1, 2, 3, 4, 6, 10, 15, 20, 30, and 50. We repeat every testing process for each program 100 times and record the number of outputs returning PASS. A smaller number indicates higher effectiveness in detecting bugs. Throughout the experiment, we aim to identify a set of $k$ and $\epsilon$ values that balance between accuracy and running time.

\vspace*{1mm}
\noindent
$\bullet$ \textbf{RQ2:} 
We fix $k$ and $\epsilon$ in the EQ and UN algorithms and select different values of $s$ to run the algorithms for each expected-pass program. We calculate the value of $s_0$ using Proposition~\ref{prop:SelectT}, and then choose $s = 0.05s_0$, $0.1s_0$, $0.2s_0$, $0.3s_0$, $0.4s_0$, $0.5s_0$, $0.7s_0$, $0.9s_0$, and $1.0s_0$, respectively. We repeat the testing process for each program with different choices of $s$ 100 times respectively, and record the number of outputs that return PASS. A higher number indicates fewer misjudgments by the algorithm.

\vspace*{1mm}
\noindent
$\bullet$ {\textbf{RQ3:} We conducted experiments to evaluate the performance of the optimized algorithms - Algorithm~\ref{Alg:OptEqCheck} and Algorithm~\ref{Alg:OptUnCheck}. We keep the parameters $k$, $\epsilon$, and $s$ fixed and run Algorithms~\ref{Alg:OptEqCheck} and~\ref{Alg:OptUnCheck} on both the expected-pass and expected-fail programs. The values of $k$ and $\epsilon$ are selected based on the balanced configuration determined in the experiments for RQ1, while the value of $t$ is determined using Proposition~\ref{prop:SelectT}. We explored the impact of different values of $t$ on the algorithm's performance.
We selected $t$ values as $t$ = 1, 2, 3, 4, 6, 10, 15, 20, 30, and 50, respectively. For each benchmark program and different choice of $t$, we repeated the testing process 100 times, recording the number of outputs that returned a "PASS" result. Additionally, we kept track of the number of times the optimization rules were triggered during the testing. The trigger condition was defined as Algorithm~\ref{Alg:OptEqCheck} returning at line 6 or 11, and Algorithm~\ref{Alg:OptUnCheck} returning at line 5.
By comparing the number of "PASS" results obtained from the optimized algorithms with those from the original algorithms and analyzing the number of triggers, we can gain insights into the performance of the optimized algorithms. This evaluation allows us to understand how effectively the optimization rules improve the algorithms' performance.

\vspace*{1mm}
\noindent
$\bullet$ \textbf{RQ4:}
We executed both the original and optimized algorithms for each benchmark program listed in Table~\ref{tab:benchmark}. To ensure a balanced selection of parameters $k$, $\epsilon$, and $t$, we leveraged the findings from the experiments conducted in RQ1 and RQ3. Additionally, we calculated the value of $s$ using Proposition~\ref{prop:SelectT}. Each program was tested 100 times, and we recorded the number of outputs that returned PASS. If an expected-pass program yielded a number close to 100 and an expected-fail program yielded a number close to 0, we can conclude that the algorithm is effective for them. Moreover, we measured the running time of both the original algorithms ($T_0$) and the optimized algorithms ($T_{opt}$) for each program. The ratio $T_{opt}/T_0$ serves as a measure of the efficiency of the optimization rules, where a smaller ratio indicates higher efficiency.

\vspace*{1mm}
\noindent
$\bullet$ \textbf{RQ5:} We identify specific cases from RQ1 to RQ4 where our algorithms underperform. Our subsequent analysis aims to extract insights regarding the bug patterns within these particular instances.

%\sout{We use Q\# language (SDK version 0.26.233415) and its simulator to conduct our experiments with the support of our checking tool QBlackBoxChecker. We conduct experiments on a personal computer with an Intel Core i7-1280P CPU and 16 GB RAM. For the experiments involving running time, we run the testing tasks one by one to avoid measurement errors from CPU frequency reduction by computer cooling limitations.}

\subsection{Result Analysis}
\label{subsec:results}

\subsubsection{RQ1: About parameter $k$ and $\epsilon$}
\label{subsubsec:RQ1}
%\vspace*{1mm}
%\noindent
%$\bullet$ \textbf{RQ1:}

The testing results for benchmark programs with different parameters $k$ and $\epsilon$ for EQ and UN are presented in Table~\ref{tab:EQkepsilon} and~\ref{tab:UNkepsilon}, while the results with different $k$ for ID are shown in Table~\ref{tab:IDkepsilon}. In Table~\ref{tab:EQkepsilon}, "Test No. 5.1" means the 1st program/pair of No. 5 in Table~\ref{tab:benchmark}. Similarly, the same notation is used in Table~\ref{tab:IDkepsilon} and ~\ref{tab:UNkepsilon}.
Most of the results regarding the number of PASS are 0, indicating that our algorithms effectively detect failures in the expected-fail programs/pairs corresponding to the selected parameters. Notably, even a single gate or measurement mutation can lead to a significant deviation from the original program. While some programs or pairs have non-zero results, their trend suggests that the number of PASS decreases as $k$ increases and $\epsilon$ decreases. This indicates that these programs or pairs are closer to the correct version, and the selected parameters are too lenient to identify their bugs.

According to Corollary~\ref{cor:t}, the time complexity of EQ and UN is more sensitive to $\epsilon$ (which depends on $\frac{1}{\epsilon^2}$) than to $k$ (which depends on $k\log k$). Therefore, it is advisable to avoid choosing very small values for $\epsilon$, but $k$ can be slightly larger. Interestingly, for some programs such as No. 6.2 and No. 27.3, even adopting excessively small $\epsilon$, the number of unrevealed bugs cannot be reduced to a low level. Based on the findings in Tables~\ref{tab:EQkepsilon}, \ref{tab:IDkepsilon} and~\ref{tab:UNkepsilon}, we proceed with $k=4$ and $\epsilon=0.15$ for EQ and UN, and $k=50$ for ID in the subsequent experiments.

\begin{table*}
\centering
\caption{The number of PASS with different $k$ and $\epsilon$ for expected-fail programs in EQ.}
\label{tab:EQkepsilon}
\resizebox*{\linewidth}{!}{
\begin{tabular}{c|c|cccccc|cccccc|cccccc|cccccc}
	\toprule
	\multirow{4}{*}{\textbf{\textbf{Task}}} && \multicolumn{6}{c|}{$\epsilon=0.05$} & \multicolumn{6}{c|}{$\epsilon=0.10$} & \multicolumn{6}{c|}{$\epsilon=0.15$} & \multicolumn{6}{c}{$\epsilon=0.20$}\\
	\cline{2-26}
	& \diagbox{ \makecell{\textbf{Test} \\ \textbf{No.}} }{$k$} & 1 & 2 & 3 & 4 & 6 & 10 & 1 & 2 & 3 & 4 & 6 & 10 & 1 & 2 & 3 & 4 & 6 & 10 & 1 & 2 & 3 & 4 & 6 & 10\\
	\midrule
	
	 \multirow{22}{*}{EQ} & 5.1 & 35 & 10 & 4 & 2 & 0 & 0 & 31 & 10 & 4 & 1 & 0 & 0 & 30 & 9 & 5 & 2 & 0 & 0 & 44 & 18 & 3 & 2 & 0 & 0\\
	 & 5.2 & 0 & 0 & 0 & 0 & 0 & 0 & 0 & 0 & 0 & 0 & 0 & 0 & 0 & 0 & 0 & 0 & 0 & 0 & 0 & 0 & 0 & 0 & 0 & 0\\
	 & 5.3 & 16 & 4 & 0 & 0 & 0 & 0 & 12 & 7 & 1 & 1 & 0 & 0 & 18 & 7 & 1 & 1 & 0 & 0 & 28 & 5 & 1 & 0 & 0 & 0\\
	 & 5.4 & 23 & 5 & 2 & 2 & 0 & 0 & 19 & 6 & 0 & 0 & 0 & 0 & 17 & 5 & 0 & 0 & 0 & 0 & 23 & 2 & 1 & 0 & 0 & 0\\
	 & 5.5 & 33 & 12 & 6 & 0 & 0 & 0 & 21 & 14 & 4 & 2 & 0 & 0 & 38 & 15 & 3 & 1 & 0 & 0 & 38 & 9 & 3 & 1 & 1 & 0\\
	 \cmidrule{2-26}
	 & 6.1 & 10 & 1 & 1 & 0 & 0 & 0 & 9 & 4 & 0 & 0 & 0 & 0 & 17 & 1 & 0 & 0 & 0 & 0 & 17 & 4 & 0 & 0 & 0 & 0\\
	 & 6.2 & 58 & 30 & 25 & 11 & 4 & 2 & 53 & 30 & 17 & 14 & 4 & 1 & 81 & 65 & 46 & 31 & 16 & 12 & 78 & 66 & 54 & 58 & 36 & 21\\
	 \cmidrule{2-26}
	 & 7.1 & 53 & 40 & 12 & 7 & 3 & 0 & 63 & 34 & 24 & 12 & 5 & 0 & 75 & 58 & 48 & 35 & 21 & 10 & 88 & 75 & 59 & 66 & 46 & 22\\
	 & 7.2 & 33 & 11 & 6 & 0 & 0 & 0 & 40 & 12 & 9 & 1 & 0 & 0 & 55 & 34 & 12 & 6 & 0 & 1 & 62 & 47 & 24 & 23 & 7 & 1\\
	 & 7.3 & 25 & 2 & 0 & 0 & 0 & 0 & 25 & 5 & 0 & 0 & 0 & 0 & 20 & 4 & 1 & 0 & 0 & 0 & 20 & 8 & 2 & 1 & 0 & 0\\
	 \cmidrule{2-26}
	 & 8.1 & 0 & 0 & 0 & 0 & 0 & 0 & 0 & 0 & 0 & 0 & 0 & 0 & 0 & 0 & 0 & 0 & 0 & 0 & 0 & 0 & 0 & 0 & 0 & 0\\
	 & 8.2 & 21 & 14 & 4 & 0 & 0 & 0 & 26 & 15 & 6 & 0 & 0 & 0 & 30 & 12 & 2 & 1 & 0 & 0 & 44 & 12 & 7 & 2 & 0 & 0\\
	 & 8.3 & 0 & 0 & 0 & 0 & 0 & 0 & 1 & 0 & 0 & 0 & 0 & 0 & 8 & 1 & 0 & 0 & 0 & 0 & 31 & 10 & 2 & 0 & 0 & 0\\
	 & 8.4 & 1 & 0 & 0 & 0 & 0 & 0 & 1 & 0 & 0 & 0 & 0 & 0 & 0 & 0 & 0 & 0 & 0 & 0 & 1 & 0 & 0 & 0 & 0 & 0\\
	 & 8.5 & 2 & 0 & 0 & 0 & 0 & 0 & 1 & 0 & 0 & 0 & 0 & 0 & 1 & 0 & 0 & 0 & 0 & 0 & 3 & 0 & 0 & 0 & 0 & 0\\
	 \cmidrule{2-26}
	 & 9.1 & 4 & 0 & 0 & 0 & 0 & 0 & 1 & 0 & 0 & 0 & 0 & 0 & 2 & 0 & 0 & 0 & 0 & 0 & 2 & 0 & 0 & 0 & 0 & 0 \\
	 & 9.2 & 2 & 0 & 0 & 0 & 0 & 0 & 4 & 0 & 0 & 0 & 0 & 0 & 3 & 0 & 0 & 0 & 0 & 0 & 3 & 0 & 0 & 0 & 0 & 0 \\
	 & 9.3 & 0 & 0 & 0 & 0 & 0 & 0 & 0 & 0 & 0 & 0 & 0 & 0 & 0 & 0 & 0 & 0 & 0 & 0 & 0 & 0 & 0 & 0 & 0 & 0 \\
	 & 9.4 & 0 & 0 & 0 & 0 & 0 & 0 & 0 & 0 & 0 & 0 & 0 & 0 & 0 & 0 & 0 & 0 & 0 & 0 & 0 & 0 & 0 & 0 & 0 & 0\\
	 & 9.5 & 0 & 0 & 0 & 0 & 0 & 0 & 0 & 0 & 0 & 0 & 0 & 0 & 0 & 0 & 0 & 0 & 0 & 0 & 2 & 0 & 0 & 0 & 0 & 0\\
	
	\bottomrule
\end{tabular}
}
\end{table*}

\begin{table}
\centering
\caption{The number of PASS with different $k$ for expected-fail programs in ID.}
\label{tab:IDkepsilon}
\resizebox{\linewidth}{!}{
\begin{tabular}{c|c|cccccccccc}
	\toprule
	\textbf{Task} & \diagbox{\makecell{\textbf{Test} \\ \textbf{No.}}}{$k$} & 1 & 2 & 3 & 4 & 6 & 10 & 15 & 20 & 30 & 50\\
	\midrule
	\multirow{18}{*}{ID} &14.1 & 67 & 46 & 33 & 22 & 5 & 0 & 0 & 0 & 0 & 0\\
	&14.2 & 96 & 87 & 76 & 64 & 55 & 28 & 18 & 17 & 6 & 0\\
	\cmidrule{2-12}
	&15.1 & 26 & 3 & 0 & 0 & 0 & 0 & 0 & 0 & 0 & 0\\
	&15.2 & 27 & 5 & 1 & 0 & 0 & 0 & 0 & 0 & 0 & 0\\
	&15.3 & 1 & 0 & 0 & 0 & 0 & 0 & 0 & 0 & 0 & 0\\
	&15.4 & 3 & 0 & 0 & 0 & 0 & 0 & 0 & 0 & 0 & 0\\
	&15.5 & 28 & 5 & 1 & 0 & 0 & 0 & 0 & 0 & 0 & 0\\
	\cmidrule{2-12}
	&16.1 & 13 & 1 & 1 & 0 & 0 & 0 & 0 & 0 & 0 & 0\\
	&16.2 & 22 & 7 & 1 & 0 & 0 & 0 & 0 & 0 & 0 & 0\\
	&16.3 & 44 & 11 & 4 & 3 & 0 & 0 & 0 & 0 & 0 & 0\\
	&16.4 & 49 & 21 & 8 & 5 & 0 & 0 & 0 & 0 & 0 & 0\\
	&16.5 & 4 & 0 & 0 & 0 & 0 & 0 & 0 & 0 & 0 & 0\\
	\cmidrule{2-12}
	&17.1 & 2 & 0 & 0 & 0 & 0 & 0 & 0 & 0 & 0 & 0\\
	&17.2 & 17 & 1 & 0 & 0 & 0 & 0 & 0 & 0 & 0 & 0\\
	&17.3 & 18 & 3 & 1 & 0 & 0 & 0 & 0 & 0 & 0 & 0\\
	&17.4 & 10 & 3 & 0 & 0 & 0 & 0 & 0 & 0 & 0 & 0\\
	&17.5 & 14 & 0 & 0 & 0 & 0 & 0 & 0 & 0 & 0 & 0\\
	\bottomrule
\end{tabular}
}
\end{table}

\begin{table*}
\centering
\caption{The number of PASS with different $k$ and $\epsilon$ for expected-fail programs in UN.}
\label{tab:UNkepsilon}
\resizebox*{\linewidth}{!}{
\begin{tabular}{c|c|cccccc|cccccc|cccccc|cccccc}
	\toprule
	\multirow{4}{*}{\textbf{\textbf{Task}}} && \multicolumn{6}{c|}{$\epsilon=0.05$} & \multicolumn{6}{c|}{$\epsilon=0.10$} & \multicolumn{6}{c|}{$\epsilon=0.15$} & \multicolumn{6}{c}{$\epsilon=0.20$}\\
	\cline{2-26}
	& \diagbox{ \makecell{\textbf{Test} \\ \textbf{No.}} }{$k$} & 1 & 2 & 3 & 4 & 6 & 10 & 1 & 2 & 3 & 4 & 6 & 10 & 1 & 2 & 3 & 4 & 6 & 10 & 1 & 2 & 3 & 4 & 6 & 10\\
	\midrule
		
	 \multirow{12}{*}{UN} & 24 & - & 0 & 0 & 0 & 0 & 0 & - & 0 & 0 & 0 & 0 & 0 & - & 0 & 0 & 0 & 0 & 0 & - & 0 & 0 & 0 & 0 & 0\\
	 \cmidrule{2-26}
	 & 25 & - & 0 & 0 & 0 & 0 & 0 & - & 0 & 0 & 0 & 0 & 0 & - & 0 & 0 & 0 & 0 & 0 & - & 0 & 0 & 0 & 0 & 0\\
	\cmidrule{2-26}
	 & 26 & - & 0 & 0 & 0 & 0 & 0 & - & 0 & 0 & 0 & 0 & 0 & - & 0 & 0 & 0 & 0 & 0 & - & 0 & 0 & 0 & 0 & 0 \\
	 \cmidrule{2-26}
	 & 27.1 & - & 0 & 0 & 0 & 0 & 0 & - & 0 & 0 & 0 & 0 & 0 & - & 0 & 0 & 0 & 0 & 0 & - & 0 & 0 & 0 & 0 & 0\\
	 & 27.2 & - & 0 & 0 & 0 & 0 & 0 & - & 0 & 0 & 0 & 0 & 0 & - & 0 & 0 & 0 & 0 & 0 & - & 0 & 0 & 0 & 0 & 0\\
	 & 27.3 & - & 77 & 61 & 61 & 43 & 26 & - & 78 & 66 & 54 & 50 & 28 & - & 81 & 64 & 67 & 48 & 30 & - & 87 & 83 & 81 & 74 & 53 \\
	 & 27.4 & - & 0 & 0 & 0 & 0 & 0 & - & 0 & 0 & 0 & 0 & 0 & - & 0 & 0 & 0 & 0 & 0 & - & 0 & 0 & 0 & 0 & 0\\
	 & 27.5 & - & 0 & 0 & 0 & 0 & 0 & - & 0 & 0 & 0 & 0 & 0 & - & 0 & 0 & 0 & 0 & 0 & - & 0 & 0 & 0 & 0 & 0\\
	 \cmidrule{2-26}
	 & 28.1 & - & 6 & 2 & 1 & 0 & 0 & - & 6 & 2 & 1 & 0 & 0 & - & 5 & 1 & 0 & 0 & 0 & - & 8 & 3 & 2 & 1 & 0\\
	 & 28.2 & - & 4 & 0 & 0 & 1 & 0 & - & 7 & 0 & 0 & 0 & 0 & - & 11 & 1 & 0 & 0 & 0 & - & 5 & 4 & 0 & 0 & 0 \\
	
	\bottomrule
\end{tabular}
}
\end{table*}
%%%%%%%%%%%%%%%%%%%%%%%%%%%%%%%%%%%%%%%%%%%%%%%%%%%%%%%%%%%

\subsubsection{RQ2: About parameter $s$}
\label{subsubsec:RQ2}
%\vspace*{1mm}
%\noindent
%$\bullet$ \textbf{RQ2:}
The testing results for benchmark programs with different parameters $s$ for EQ and UN are presented in Figure~\ref{fig:EqUnT}. It is evident that selecting $s=s_0$ leads to nearly all expected-pass programs returning PASS. This outcome demonstrates the effectiveness of the chosen parameter $s$ based on Proposition~\ref{prop:SelectT} in controlling type II errors. Furthermore, it is notable that the bad results (i.e., the rate of PASS is less than $1-\alpha_2 = 0.9$) is primarily observed in certain programs when $s<0.5s_0$. Additionally, for Test No. 1, 3, and 4, the number of PASS remains at 100 even when $s=0.05s_0$. This observation indicates that different programs exhibit varying degrees of sensitivity to the parameter $s$, and the selection of $s$ in Proposition~\ref{prop:SelectT} is conservative to ensure a high PASS rate in the worst-case scenario.

\begin{figure*}
\centering
\caption{The number of PASS with different $s$ for expected-pass programs in EQ and UN under $k=4, \epsilon=0.15$}
\label{fig:EqUnT}
\includegraphics[scale=0.6]{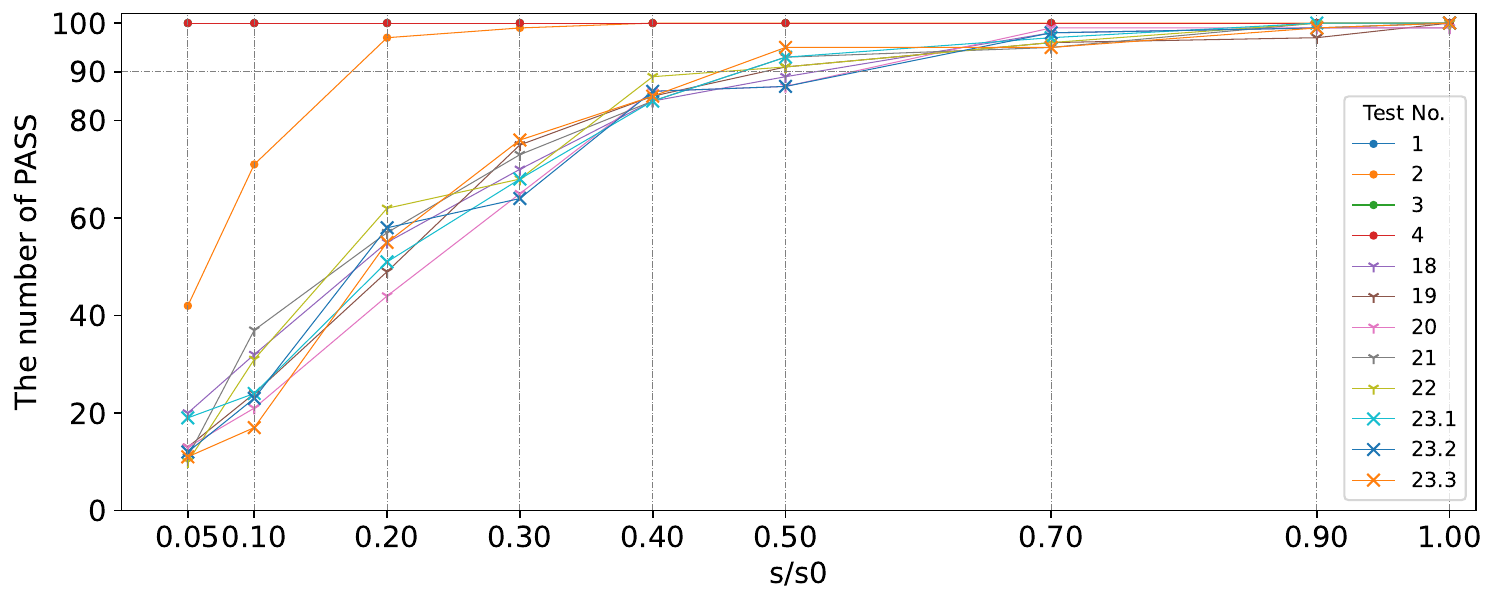}
\end{figure*}

\subsubsection{RQ3: The performance of optimized algorithms}
\label{subsubsec:RQ3}

The testing results with different $t$ for the optimized algorithms are presented in Table~\ref{tab:TTrace}. It can be observed that most of the expected-fail programs are killed by the optimization rules, indicating the effectiveness of our optimization rules in accelerating bug detection. In addition, most expected-pass programs are not triggered by the optimization rules. It is because our optimization rules only relate to returning FAIL immediately.

Several interesting cases, namely No. 2, No. 6.2, No. 7.1, No. 26, and No. 27.3, deserve attention. For No. 2, it is an expected-pass case but triggered by optimization rules, which means a misjudgment. As $t$ increases, the number of PASS instances increases while the number of triggers decreases. It is noteworthy that \texttt{Cir2A} and \texttt{Cir2B} contain measurements (Figure~\ref{fig:TwoCir}(b)), resulting in expected mixed-state outputs. Therefore, a small value of $t$ may lead to the misjudgment of pure and mixed states, increasing the likelihood of misjudgment of results in this particular test scenario. In the case of No. 6.2 and No. 7.1, it appears that a small $t$ effectively excludes incorrect results. However, it is important to observe that as $t$ increases, the number of triggers decreases. This implies that returning FAIL based on optimization rules is also a misjudgment (it is similar to No. 2). No. 26 (program \texttt{Reset}) is an intriguing case as it represents an expected-fail scenario that cannot be triggered by the optimization rule of UN. This is a typical example that shows that purity preservation alone is insufficient for unitarity checking. Lastly, for No. 27.3, a significant observation is that the number of PASS instances in the optimized algorithm is considerably smaller than in the original algorithm (column 'Base'). It should be noted that the original algorithm is based on orthogonal preservation, while the optimization rule is based on purity preservation. This suggests that purity preservation checking is more effective than orthogonal preservation checking in this particular task.

Based on the results presented in Table~\ref{tab:TTrace}, we find that $t=20$ provides a balanced choice, which we will further analyze in the subsequent overall evaluation.

\subsubsection{RQ4: Overall performance for benchmark programs}
\label{subsubsec:RQ4}
%\vspace*{1mm}
%\noindent
%$\bullet$ \textbf{RQ3:}
The running results are shown in Table~\ref{tab:results}, including both the original and optimized algorithms. We can see that our methods work well for most programs, i.e., the percentage of PASS is near 100 for expected-pass programs and near 0 for expected-fail programs. It means that our methods are effective for most benchmark programs with the parameters setting in the "Running Parameters" column of Table~\ref{tab:results}. The "Average Run Time" column gives the average running time in a single test.

The average running time of expected-fail programs is less than that of corresponding expected-pass programs. That is because the FAIL result is returned as long as one test point fails. Conversely, the PASS result is returned only if all test points pass. ID is faster than EQ and UN because ID does not have the repeat of the Swap Test in EQ and UN. For example, consider programs No. 4 and No. 11, which both check the relation \texttt{QFT$\circ$invQFT} $=I$. No. 4 uses EQ and another identity program, while No. 11 uses ID directly to check it. No. 11 requires only about one-thousandth (for the original EQ algorithm) or one-eighth (for optimized EQ algorithm) running time of No. 4. So identity checking can be implemented efficiently. It also tells us that finding more quantum metamorphic relations which can be reduced to identity checking is valuable.

Take note of the column labeled "$T_{opt}/T_o$," which indicates the efficiency of the optimization rules. We observe that for the equivalence (EQ) checking, the optimized algorithm demonstrates varying degrees of acceleration with the same parameter configuration. The acceleration ranges from dozens of percent (e.g., No. 2, No. 6.2, No. 7.1, No. 7.2, No. 7.3) to several hundred multiples (remaining EQ cases). These results indicate the effectiveness of the optimization rules for EQ. In the case of unitarity (UN) checking, the optimized algorithm is slightly slower than the original algorithm for all expected-pass cases and No.26. This is due to the additional purity preservation checking in the optimized algorithm, and they cannot be excluded by purity preservation checking.
Consequently, the orthogonal checking cannot be bypassed. Nevertheless, the optimized algorithm significantly accelerates several hundred multiples for other UN cases. Notably, the optimization may also enhance the accuracy of bug detection in some cases, such as No.27.3.

\begin{table*}
\centering
\caption{The number of PASS and triggers by optimization rules with different $t$ for programs in EQ and UN under $k=4$, $\epsilon=0.15$ and $s$ selected by Proposition~\ref{prop:SelectT}}
\label{tab:TTrace}
\resizebox*{\linewidth}{!}{
\begin{tabular}{c|c|c|c|cccccccccc|cccccccccc}
\toprule
&&& \multicolumn{11}{c|}{\textbf{Number of PASS}} & \multicolumn{10}{c}{ \textbf{Number of triggers by optimization rules} } \\
\cline{3-24}
\textbf{Task} & \textbf{Expected} & \diagbox{ \makecell{\textbf{Test} \\ \textbf{No.}} }{$t$} & Base & 1 & 2 & 3 & 4 & 6 & 10 & 15 & 20 & 30 & 50 & 1 & 2 & 3 & 4 & 6 & 10 & 15 & 20 & 30 & 50 \\
\midrule

	\multirow{25}{*}{EQ} & \multirow{4}{*}{PASS} & 1 & 100 & 100 & 100 & 100 & 100 & 100 & 100 & 100 & 100 & 100 & 100 & 0 & 0 & 0 & 0 & 0 & 0 & 0 & 0 & 0 & 0 \\
	&& 2 & 100 & 4 & 3 & 8 & 19 & 38 & 84 & 94 & 99 & 100 & 100 & 96 & 97 & 92 & 81 & 62 & 16 & 6 & 1 & 0 & 0 \\
	&& 3 & 100 & 100 & 100 & 100 & 100 & 100 & 100 & 100 & 100 & 100 & 100 & 0 & 0 & 0 & 0 & 0 & 0 & 0 & 0 & 0 & 0 \\
	&& 4 & 100 & 100 & 100 & 100 & 100 & 100 & 100 & 100 & 100 & 100 & 100 & 0 & 0 & 0 & 0 & 0 & 0 & 0 & 0 & 0 & 0 \\
	
	\cmidrule{2-24}
	
	& \multirow{20}{*}{FAIL} & 5.1 & 2 & 23 & 4 & 3 & 1 & 2 & 0 & 2 & 3 & 1 & 1 & 77 & 96 & 97 & 99 & 98 & 100 & 98 & 97 & 99 & 99 \\
	&& 5.2 & 0 & 17 & 7 & 1 & 0 & 0 & 0 & 0 & 0 & 0 & 0 & 83 & 93 & 99 & 100 & 100 & 100 & 100 & 100 & 100 & 100 \\
	&& 5.3 & 1 & 23 & 7 & 1 & 0 & 0 & 0 & 0 & 1 & 0 & 0 & 77 & 93 & 99 & 100 & 100 & 100 & 100 & 99 & 100 & 100 \\
	&& 5.4 & 0 & 24 & 12 & 4 & 2 & 0 & 0 & 1 & 0 & 0 & 0 & 76 & 88 & 96 & 98 & 100 & 100 & 99 & 100 & 100 & 100 \\
	&& 5.5 & 1 & 30 & 12 & 4 & 2 & 3 & 1 & 0 & 1 & 4 & 2 & 70 & 88 & 96 & 98 & 97 & 99 & 100 & 99 & 96 & 98 \\
	\cline{3-24}
	&& 6.1 & 0 & 5 & 0 & 0 & 0 & 0 & 0 & 0 & 0 & 0 & 0 & 95 & 100 & 100 & 100 & 100 & 100 & 100 & 100 & 100 & 100 \\
	&& 6.2 & 31 & 3 & 6 & 3 & 11 & 14 & 33 & 42 & 42 & 42 & 33 & 93 & 91 & 83 & 67 & 48 & 18 & 3 & 0 & 0 & 0 \\
	\cline{3-24}
	&& 7.1 & 35 & 4 & 7 & 5 & 7 & 25 & 30 & 38 & 39 & 43 & 34 & 92 & 84 & 77 & 70 & 42 & 19 & 4 & 0 & 0 & 0 \\
	&& 7.2 & 6 & 2 & 2 & 0 & 2 & 6 & 7 & 18 & 8 & 10 & 12 & 91 & 87 & 70 & 66 & 41 & 14 & 5 & 0 & 0 & 0 \\
	&& 7.3 & 0 & 4 & 0 & 0 & 0 & 1 & 0 & 0 & 0 & 0 & 0 & 81 & 77 & 69 & 46 & 44 & 19 & 19 & 16 & 18 & 13 \\
	\cline{3-24}
	&& 8.1 & 0 & 2 & 0 & 0 & 0 & 0 & 0 & 0 & 0 & 0 & 0 & 98 & 100 & 100 & 100 & 100 & 100 & 100 & 100 & 100 & 100 \\
	&& 8.2 & 1 & 20 & 14 & 7 & 3 & 1 & 0 & 2 & 2 & 2 & 1 & 80 & 86 & 93 & 97 & 99 & 100 & 98 & 98 & 98 & 99 \\
	&& 8.3 & 0 & 14 & 0 & 0 & 0 & 0 & 0 & 0 & 0 & 0 & 0 & 86 & 99 & 100 & 100 & 100 & 100 & 100 & 100 & 100 & 100 \\
	&& 8.4 & 0 & 2 & 0 & 0 & 0 & 0 & 0 & 0 & 0 & 0 & 0 & 98 & 100 & 100 & 100 & 100 & 100 & 100 & 100 & 100 & 100 \\
	&& 8.5 & 0 & 0 & 0 & 0 & 0 & 0 & 0 & 0 & 0 & 0 & 0 & 100 & 100 & 100 & 100 & 100 & 100 & 100 & 100 & 100 & 100 \\
	\cline{3-24}
	&& 9.1 & 0 & 12 & 1 & 0 & 0 & 0 & 0 & 0 & 0 & 0 & 0 & 88 & 99 & 100 & 100 & 100 & 100 & 100 & 100 & 100 & 100 \\
	&& 9.2 & 0 & 16 & 2 & 0 & 0 & 0 & 0 & 0 & 0 & 0 & 0 & 84 & 98 & 100 & 100 & 100 & 100 & 100 & 100 & 100 & 100 \\
	&& 9.3 & 0 & 8 & 0 & 0 & 0 & 0 & 0 & 0 & 0 & 0 & 0 & 92 & 100 & 100 & 100 & 100 & 100 & 100 & 100 & 100 & 100 \\
	&& 9.4 & 0 & 6 & 0 & 0 & 0 & 0 & 0 & 0 & 0 & 0 & 0 & 94 & 100 & 100 & 100 & 100 & 100 & 100 & 100 & 100 & 100 \\
	&& 9.5 & 0 & 3 & 0 & 0 & 0 & 0 & 0 & 0 & 0 & 0 & 0 & 97 & 100 & 100 & 100 & 100 & 100 & 100 & 100 & 100 & 100 \\

\midrule

	\multirow{19}{*}{UN} & \multirow{8}{*}{PASS}& 18 & 99 & 100 & 100 & 100 & 100 & 100 & 100 & 100 & 100 & 100 & 100 & 0 & 0 & 0 & 0 & 0 & 0 & 0 & 0 & 0 & 0 \\
	&& 19 & 100 & 100 & 100 & 100 & 100 & 100 & 100 & 100 & 100 & 100 & 100 & 0 & 0 & 0 & 0 & 0 & 0 & 0 & 0 & 0 & 0 \\
	&& 20 & 100 & 100 & 100 & 100 & 100 & 100 & 100 & 100 & 100 & 100 & 100 & 0 & 0 & 0 & 0 & 0 & 0 & 0 & 0 & 0 & 0 \\
	&& 21 & 99 & 100 & 100 & 100 & 100 & 100 & 100 & 100 & 100 & 100 & 100 & 0 & 0 & 0 & 0 & 0 & 0 & 0 & 0 & 0 & 0 \\
	&& 22 & 99 & 100 & 100 & 100 & 100 & 100 & 100 & 100 & 100 & 100 & 100 & 0 & 0 & 0 & 0 & 0 & 0 & 0 & 0 & 0 & 0 \\
	\cline{3-24}
	&& 23.1 & 99 & 100 & 100 & 100 & 100 & 100 & 100 & 100 & 100 & 100 & 100 & 0 & 0 & 0 & 0 & 0 & 0 & 0 & 0 & 0 & 0 \\
	&& 23.2 & 100 & 100 & 100 & 100 & 100 & 100 & 100 & 100 & 100 & 100 & 100 & 0 & 0 & 0 & 0 & 0 & 0 & 0 & 0 & 0 & 0 \\
	&& 23.3 & 99 & 100 & 100 & 100 & 100 & 100 & 100 & 100 & 100 & 100 & 100 & 0 & 0 & 0 & 0 & 0 & 0 & 0 & 0 & 0 & 0 \\
	
	\cmidrule{2-24}
	
	& \multirow{10}{*}{FAIL} & 24 & 0 & 0 & 0 & 0 & 0 & 0 & 0 & 0 & 0 & 0 & 0 & 67 & 93 & 97 & 98 & 100 & 100 & 100 & 100 & 100 & 100 \\
	&& 25 & 0 & 0 & 0 & 0 & 0 & 0 & 0 & 0 & 0 & 0 & 0 & 69 & 87 & 96 & 99 & 100 & 100 & 100 & 100 & 100 & 100 \\
	&& 26 & 0 & 0 & 0 & 0 & 0 & 0 & 0 & 0 & 0 & 0 & 0 & 0 & 0 & 0 & 0 & 0 & 0 & 0 & 0 & 0 & 0 \\
	\cline{3-24}
	&& 27.1 & 0 & 0 & 0 & 0 & 0 & 0 & 0 & 0 & 0 & 0 & 0 & 87 & 97 & 100 & 100 & 100 & 100 & 100 & 100 & 100 & 100 \\
	&& 27.2 & 0 & 0 & 0 & 0 & 0 & 0 & 0 & 0 & 0 & 0 & 0 & 61 & 76 & 86 & 83 & 99 & 97 & 98 & 98 & 98 & 100 \\
	&& 27.3 & 67 & 22 & 11 & 4 & 2 & 1 & 0 & 0 & 0 & 0 & 0 & 60 & 84 & 96 & 96 & 99 & 100 & 100 & 100 & 100 & 100 \\
	&& 27.4 & 0 & 0 & 0 & 0 & 0 & 0 & 0 & 0 & 0 & 0 & 0 & 91 & 96 & 100 & 100 & 100 & 100 & 100 & 100 & 100 & 100 \\
	&& 27.5 & 0 & 0 & 0 & 0 & 0 & 0 & 0 & 0 & 0 & 0 & 0 & 86 & 100 & 100 & 100 & 100 & 100 & 100 & 100 & 100 & 100 \\
	\cline{3-24}
	&& 28.1 & 0 & 1 & 0 & 0 & 0 & 0 & 0 & 0 & 0 & 0 & 0 & 56 & 79 & 95 & 98 & 100 & 100 & 100 & 100 & 100 & 100 \\
	&& 28.2 & 0 & 0 & 0 & 0 & 0 & 0 & 0 & 0 & 0 & 0 & 0 & 67 & 87 & 93 & 96 & 100 & 100 & 100 & 100 & 100 & 100 \\

\bottomrule
\end{tabular}
}
\end{table*}

\begin{table*}
\centering
\caption{The testing results of benchmark programs with both original algorithms and optimized algorithms}
\label{tab:results}
\resizebox*{!}{\textheight-15mm}{
\begin{tabular}{c|c|c|c|c|cc|cc|c}
	\toprule
	&&&&& \multicolumn{2}{c|}{\textbf{Original Algorithm}} & \multicolumn{2}{c|}{\textbf{Optimized Algorithm}} & \\
	\textbf{Task} & \makecell{\textbf{Running} \\ \textbf{Parameters}} & \textbf{Expected} & \makecell{\textbf{Test} \\ \textbf{No.}} & \textbf{\#Qubits} & \textbf{\% of PASS} & \makecell{\textbf{Average} \\ \textbf{Run Time $T_0$}} & \textbf{\% of PASS} & \makecell{\textbf{Average} \\ \textbf{Run Time $T_{opt}$ }} & \textbf{$T_{opt} / T_0$} \\
	\midrule
	
	\multirow{25}{*}{EQ} & \multirow{25}{*}{ \makecell[c]{$k=4$ \\ $s=1545$ \\ $\epsilon=0.15$ \\ \\ $t = 20$ \\ for optimized \\ algorithm.} } & \multirow{4}{*}{PASS} & 1 & 2 & 100 & 3.0s & 100 & 42ms & $1.4\times 10^{-2}$ \\
	&&& 2 & 3 & 100 & 4.98s & 100 & 4.2s & 0.84 \\
	&&& 3 & 5 & 100 & 21.9s & 100 & 279ms & $1.3\times 10^{-2}$ \\
	&&& 4 & 5 & 100 & 20.5s & 100 & 283ms & $1.4\times 10^{-2}$ \\
	
	\cmidrule{3-10}
	
	&& \multirow{20}{*}{FAIL} & 5.1 & \multirow{5}{*}{2} & 2 & 1.02s & 0 & 15ms & $1.5\times 10^{-2}$ \\
	&&& 5.2 &  & 0 & 780ms & 0 & 9.7ms & $1.2\times 10^{-2}$ \\
	&&& 5.3 &  & 0 & 1.02s & 0 & 12ms & $1.2\times 10^{-2}$\\
	&&& 5.4 &  & 0 & 1.02s & 0 & 13ms & $1.3\times 10^{-2}$ \\
	&&& 5.5 &  & 0 & 1.2s & 3 & 11ms & $9.2\times 10^{-3}$ \\
	\cline{4-10}
	&&& 6.1 & \multirow{2}{*}{3} & 0 & 1.2s & 0 & 9.1ms & $7.6\times 10^{-3}$ \\
	&&& 6.2 &  & 31 & 3.36s & 40 & 3.24s & 0.96 \\
	\cline{4-10}
	&&& 7.1 & \multirow{3}{*}{3} & 45 & 3.72s & 28 & 2.76s & 0.74 \\
	&&& 7.2 &  & 14 & 2.76s & 6 & 1.86s & 0.67 \\
	&&& 7.3 &  & 0 & 1.56s & 0 & 1.32s & 0.85 \\
	\cline{4-10}
	&&& 8.1 & \multirow{5}{*}{5} & 0 & 5.34s & 0 & 28ms & $5.2\times 10^{-3}$ \\
	&&& 8.2 & & 1 & 8.22s & 1 & 70ms & $8.5\times 10^{-3}$ \\
	&&& 8.3 &  & 0 & 6.18s & 0 & 35ms & $5.7\times 10^{-3}$ \\
	&&& 8.4 &  & 0 & 5.7s & 0 & 33ms & $5.8\times 10^{-3}$ \\
	&&& 8.5 &  & 0 & 5.88s & 0 & 33ms & $5.6\times 10^{-3}$ \\
	\cline{4-10}
	&&& 9.1 & \multirow{5}{*}{5} & 0 & 4.92s & 0 & 49ms & $1.0\times 10^{-2}$ \\
	&&& 9.2 &  & 0 & 5.34s & 0 & 55ms & $1.0\times 10^{-2}$ \\
	&&& 9.3 &  & 0 & 5.16s & 0 & 52ms & $1.0\times 10^{-2}$ \\
	&&& 9.4 &  & 0 & 5.22s & 0 & 53ms & $1.0\times 10^{-2}$ \\
	&&& 9.5 &  & 0 & 5.46s & 0 & 21ms & $9.3\times 10^{-3}$ \\
	
	\midrule
	
	\multirow{21}{*}{ID} & \multirow{21}{*}{ $k=50$ } & \multirow{4}{*}{PASS} & 10 & 6 & 100 & 12ms & \multicolumn{3}{c}{ \multirow{21}{*}{N/A} } \\
	&&& 11 & 5 & 100 & 33ms & \multicolumn{3}{c}{ } \\
	&&& 12 & 5 & 100 & 48ms & \multicolumn{3}{c}{ } \\
	&&& 13 & 3 & 100 & 231ms & \multicolumn{3}{c}{ } \\
	
	\cmidrule{3-7}
	
	&& \multirow{17}{*}{FAIL} & 14.1 & \multirow{2}{*}{2} & 0 & 1.5ms & \multicolumn{3}{c}{ } \\
	&&& 14.2 &  & 0 & 2.5ms & \multicolumn{3}{c}{ } \\
	\cline{4-7}
	&&& 15.1 & \multirow{5}{*}{5} & 0 & 1.8ms & \multicolumn{3}{c}{ } \\
	&&& 15.2 &  & 0 & 1.8ms & \multicolumn{3}{c}{ } \\
	&&& 15.3 &  & 0 & 1.8ms & \multicolumn{3}{c}{ } \\
	&&& 15.4 &  & 0 & 1.8ms & \multicolumn{3}{c}{ } \\
	&&& 15.5 &  & 0 & 1.7ms & \multicolumn{3}{c}{ } \\
	\cline{4-7}
	&&& 16.1 & \multirow{5}{*}{5} & 0 & 2.2ms & \multicolumn{3}{c}{ } \\
	&&& 16.2 &  & 0 & 2.2ms & \multicolumn{3}{c}{ } \\
	&&& 16.3 &  & 0 & 2.2ms & \multicolumn{3}{c}{ } \\
	&&& 16.4 &  & 0 & 2.3ms & \multicolumn{3}{c}{ } \\
	&&& 16.5 &  & 0 & 2.2ms & \multicolumn{3}{c}{ } \\
	\cline{4-7}
	&&& 17.1 & \multirow{5}{*}{3} & 0 & 6.9ms & \multicolumn{3}{c}{ } \\
	&&& 17.2 &  & 0 & 8.2ms & \multicolumn{3}{c}{ } \\
	&&& 17.3 &  & 0 & 8.4ms & \multicolumn{3}{c}{ } \\
	&&& 17.4 &  & 0 & 7.8ms & \multicolumn{3}{c}{ } \\
	&&& 17.5 &  & 0 & 8.1ms & \multicolumn{3}{c}{ } \\
	
	\midrule
	
	\multirow{19}{*}{UN} & \multirow{19}{*}{ \makecell[c]{$k=4$ \\ $s=469$ \\ $\epsilon=0.15$ \\ \\ $t = 20$ \\ for optimized \\ algorithm.} } & \multirow{8}{*}{PASS} & 18 & 2 & 100 & 421ms & 100 & 434ms & 1.03 \\
	&&& 19 & 2 & 100 & 316ms & 98 & 329ms & 1.04 \\
	&&& 20 & 6 & 100 & 3.3s & 99 & 3.36s & 1.02 \\
	&&& 21 & 5 & 100 & 2.4s & 100 & 2.46s & 1.03 \\
	&&& 22 & 5 & 100 & 5.04s & 100 & 5.52s & 1.10 \\
	\cline{4-10}
	&&& 23.1 & \multirow{3}{*}{5} & 99 & 2.4s & 99 & 2.46s & 1.03 \\
	&&& 23.2 &  & 100 & 2.4s & 99 & 2.58s & 1.08 \\
	&&& 23.3 &  & 100 & 2.46s & 99 & 2.52s & 1.02 \\
	
	\cmidrule{3-10}
	
	&& \multirow{10}{*}{FAIL} & 24 & 3 & 0 & 139ms & 0 & 3.3ms & $2.4\times 10^{-2}$ \\
	&&& 25 & 3 & 0 & 131ms & 0 & 3.1ms & $2.4\times 10^{-2}$ \\
	&&& 26 & 6 & 0 & 1.08s & 0 & 1.14s & 1.06 \\
	\cline{4-10}
	&&& 27.1 & \multirow{5}{*}{5} & 0 & 583ms & 0 & 5.3ms & $9.1\times 10^{-3}$ \\
	&&& 27.2 &  & 0 & 600ms & 0 & 32ms & $4.8\times 10^{-2}$ \\
	&&& 27.3 &  & 74 & 2.6s & 0 & 7.6ms & $1.0\times 10^{-2}$ \\
	&&& 27.4 &  & 0 & 660ms & 0 & 5.3ms & $2.9\times 10^{-3}$ \\
	&&& 27.5 &  & 0 & 660ms & 0 & 7.2ms & $1.1\times 10^{-2}$ \\
	\cline{4-10}
	&&& 28.1 & \multirow{2}{*}{5} & 0 & 1.68s & 0 & 18ms & $1.1\times 10^{-2}$ \\
	&&& 28.2 &  & 0 & 1.68s & 0 & 18ms & $1.1\times 10^{-2}$ \\
	
	\bottomrule
\end{tabular}
}
\end{table*}

\subsubsection{RQ5: about the cases not performed so well}
\label{subsubsec:RQ5}

Table~\ref{tab:results} also presents the performance of error mutation programs. However, it is observed that some error mutation programs, such as No. 6.2, No. 7.1, and No. 27.3, demonstrate less effectiveness in terms of the number of PASS outcomes. It may be because they are close to the correct one. 
Here, we further research these cases and try to find some insights about bug patterns.

Table~\ref{tab:benchmark} and~\ref{tab:results} indicate that the troublesome cases predominantly arise from the benchmark programs \texttt{Cir2A}, \texttt{Cir2B}, and \texttt{QFT}. Notably, these programs also include instances where our algorithms perform well, such as in cases No. 6.1, No. 7.2, No. 7.3, No. 27.1, No. 27.2, No. 27.4, and No. 27.5. Therefore, contrasting the successful cases with the troublesome ones within the same benchmark programs could be insightful. Figure~\ref{fig:mutants} presents the Q\# implementation of \texttt{Cir2A}, \texttt{Cir2B}, and \texttt{QFT}, along with their mutant operators used in our evaluation.

Figure~\ref{fig:mutants}(a) illustrates the mutant operators for cases No. $6.1\sim 6.2$ and No. $7.1\sim 7.3$. This comparison reveals varying degrees of impact on equivalence by different mutant operators. For instance, altering an operation (as in No. 6.1, No. 7.2, No. 7.3) appears to have a more significant effect than swapping two operations (as in No. 6.2, No. 7.1). Additionally, it is noteworthy that the troublesome cases for EQ checking predominantly occur in non-unitary programs like \texttt{Cir2A} and \texttt{Cir2B}. This observation suggests that non-unitary programs may demand more precise EQ checking compared to unitary programs.

Figure~\ref{fig:mutants}(b) illustrates the mutant operators for cases No. $27.1\sim 27.5$. A key observation is that troublesome case 27.3 involves a mutation by adding a measurement at the end of the program, contrasting with the other cases where measurements are introduced at the beginning or within loops. This suggests that adding a measurement to a single qubit towards a program's end disrupts its unitarity slightly. Conversely, measurements placed in other positions, especially before controlled gates or within loops, appear to significantly amplify the measurement's effect, thereby greatly disrupting the program's overall unitarity.

In practice, we can select stricter parameters – specifically, choosing a smaller $\epsilon$ and a larger $k$, as depicted in Tables~\ref{tab:EQkepsilon} and~\ref{tab:UNkepsilon} – can enhance our ability to detect and address troublesome cases effectively. However, by Corollary ~\ref{cor:t}, this implies that a super-linearly longer running time is required. Therefore, balancing the ability to find bugs and the running time is also an aspect that should be considered. In most cases, choosing too precise parameters for a few troublesome error programs may be unnecessary.

\begin{figure*}
%\textbf{}
\centering
\subfigure[\texttt{Cir2A}, \texttt{Cir2B} and their mutant operators in Test No. 6 and 7.]{\includegraphics[scale=0.5]{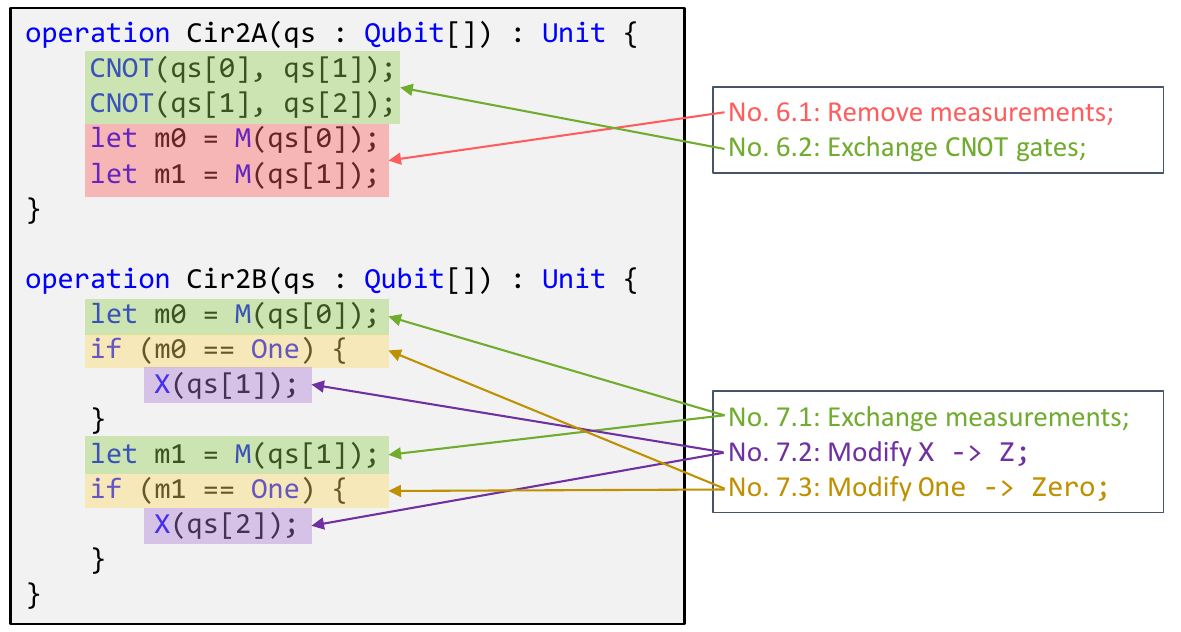}}
\\
\subfigure[\texttt{QFT} program and its mutant operators in Test No. 27.]{\includegraphics[scale=0.5]{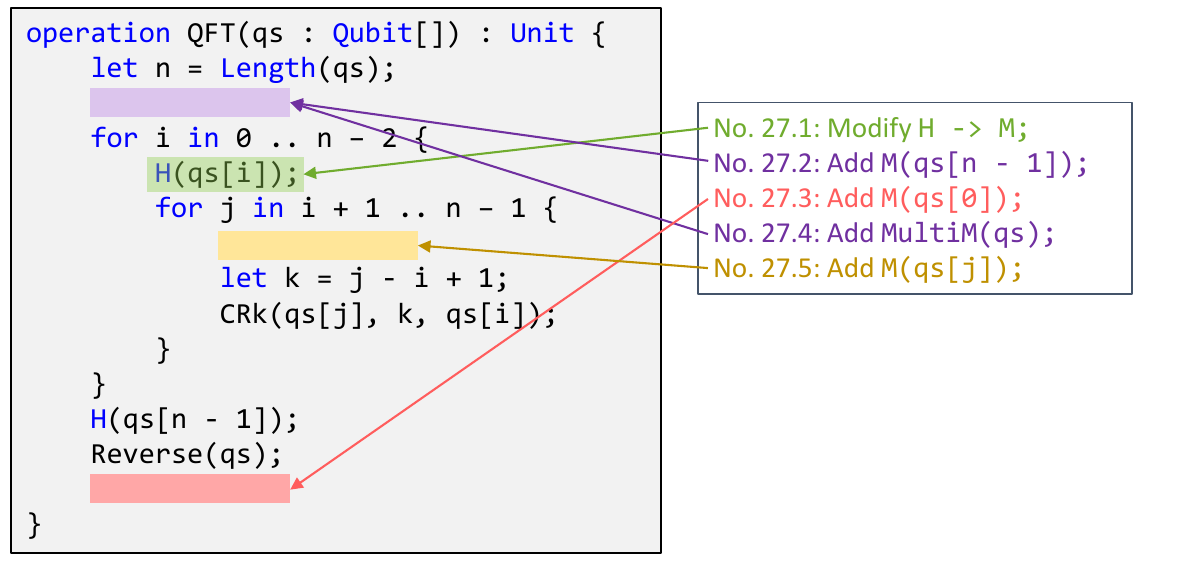}}
\caption{The Q\# implementation for programs \texttt{Cir2A}, \texttt{Cir2B}, and \texttt{QFT} and their mutant operators in the evaluation.}
\label{fig:mutants}
\end{figure*}

\section{Threats to Validity}
\label{sec:threats}

Our theoretical analysis and experimental evaluation have demonstrated the effectiveness of our methods. However, like other test methods, there are still some threats to the validity of our approach.

The first challenge arises from the limitation on the number of qubits. Our equivalence and unitarity checking methods rely on the Swap Test, which requires $2n+1$ qubits for a program with $n$ qubits. Consequently, our methods can only be applied to programs that utilize less than half of the available qubits. Furthermore, when conducting testing on a classical PC simulator, the running time for large-scale programs becomes impractical.

Although we evaluated our methods on a Q\# simulator, it is important to note that our algorithms may encounter difficulties when running on real quantum hardware. For instance, our methods rely on the availability of arbitrary gates, whereas some quantum hardware platforms impose restrictions on the types of quantum gates, such as only supporting the controlled gate (CNOT) on two adjacent qubits. Additionally, the presence of noise in the hardware can potentially diminish the effectiveness of our methods on real quantum hardware.

\section{Related Work}
\label{sec:relatedwork}

We provide an overview of the related work in the field of testing quantum programs. Quantum software testing is a nascent research area still in its early stages~\cite{miranskyy2019testing,miranskyy2021testing,garcia2021quantum,zhao2020quantum}. Various methods and techniques have been proposed for testing quantum programs from different perspectives, including quantum assertion~\cite{honarvar2020property,li2020projection}, search-based testing~\cite{wang2021generating,wang2022qusbt,wang2022mutation}, fuzz testing~\cite{wang2018quanfuzz}, combinatorial testing~\cite{wang2021application}, and metamorphic testing~\cite{abreu2022metamorphic}. These methods aim to adapt well-established classical software testing techniques to the domain of quantum computing. Moreover, Ali et al.~\cite{ali2021assessing,wang2021quito} defined input-output coverage criteria for quantum program testing and employed mutation analysis to evaluate the effectiveness of these criteria. However, their coverage criteria have limitations that restrict their applicability to testing complex, multi-subroutine quantum programs. Long and Zhao~\cite{long2022process} presented a framework for testing quantum programs with multiple subroutines. Furthermore, researchers have explored the application of mutation testing and analysis techniques to the field of quantum computing to support the testing of quantum programs~\cite{fortunato2022mutation-1,fortunato2022mutation,fortunato2022qmutpy,mendiluze2021muskit}. However, black-box testing for quantum programs has not been adequately addressed. In this paper, we propose novel checking algorithms for conducting equivalence, identity, and unitarity checking in black-box testing scenarios.

The field of equivalence checking for quantum circuits is currently active, with a number of tools dedicated to this task~\cite{viamontes2007checking,yamashita2010fast,burgholzer2020advanced,hong2021approximate,hong2022eqchk_dynamic,wang2022eqchk_seq}. Typically, this task assumes prior knowledge of the structures of two quantum circuits to verify their equivalence, known as white-box checking. Our research, however, pivots towards the equally important but less explored domain of black-box checking for equivalence.

To the best of our knowledge, our work represents the first attempt to adapt quantum information methods to support black-box testing of quantum programs. Various methods have been proposed to address problems related to quantum systems and processes, including quantum distance estimation~\cite{Flammia2011FewPauli, Cerezo2020VariationalQF}, quantum discrimination~\cite{Stephen2008Discrimination, zhang2006300discri_mixed, Zhang2007discri_puremix}, quantum tomography~\cite{tomography1989, Chuang1996statetomo}, and Swap Test~\cite{buhrman2001quantum, ekert2002direct}. By leveraging these methods, it is possible to conduct equivalence checking by estimating the distance between two quantum operations. Additionally, Kean et al.~\cite{kean2022unitarity} have conducted a thorough investigation into the unitarity estimation of quantum channels. However, it should be noted that these methods are primarily tailored for quantum information processing, specifically addressing aspects such as quantum noise. Their prerequisites and methodologies are not directly applicable to the black-box testing of quantum programs.

\section{Conclusion}
\label{sec:conclusion}

This paper presents novel methods for checking the equivalence, identity, and unitarity of quantum programs, which play a crucial role in facilitating black-box testing of these programs. Through a combination of theoretical analysis and experimental evaluations, we have demonstrated the effectiveness of our methods. The evaluation results clearly indicate that our approaches successfully enable equivalence, identity, and unitarity checking, thereby supporting the black-box testing of quantum programs.

Several areas merit further research. First, exploring additional quantum relations and devising new checking methods for them holds promise for enhancing the scope and applicability of our approach. Second, considering quantum programs that involve both classical and quantum inputs and outputs presents an intriguing research direction, necessitating the development of appropriate modeling techniques.

\section*{Acknowledgement}
\label{sec:acks}
The authors would like to thank Kean Chen for his valuable discussions regarding the writing of this paper. This work is supported in part by the National Natural Science Foundation of China under Grant 61832015 and JSPS KAKENHI Grant No. JP23H03372.

%\newpage
\bibliographystyle{acm}
\bibliography{ref}

\clearpage
\section*{Appendix}

\appendix
\section{Proof of Theorem~\ref{TheoremUcheck}}
\label{apdsec:ProofUcheck}

To prove Theorem~\ref{TheoremUcheck}, we need the following lemma first.

\begin{lemma}
Let $\mathcal{E}$ be a quantum operation on $d$-dimensional Hilbert space. Suppose it has the \textit{operator-sum} representation $\mathcal{E}(\rho)=\sum_{i=1}^{d^2}E_i \rho E_i^{\dagger}$. $\mathcal{E}$ represents a unitary transform $U$ if and only if $\forall E_i$, $E_i=k_i U$, where $k_i$ are complex numbers satisfying $\sum_{i=1}^{d^2}|k_i|^2=1$.
\end{lemma}
\begin{proof}
It is obvious that $\{F_i\}$, where $F_1=U$, $F_k=0$ for $k>1$, is an operator-sum representation of $U$. According to Theorem 8.2 in~\cite{nielsen2002quantum}, other group of elements $\{E_i\}$ represents $U$ if and only if $E_i = \sum_j w_{ij}F_j = w_{i1}U$, where all $w_{ij}$ constitutes a unitary matrix and thus $\forall j, \sum_i |w_{ij}|^2=1$. Let $k_i = w_{i1}$ and then $E_i=k_i U$.
\end{proof}

Now we are able to prove Theorem~\ref{TheoremUcheck}.

\begin{proof}[for \textbf{Theorem~\ref{TheoremUcheck}}]\

"$\Rightarrow$." Note that if $\rho_1$, $\rho_2$ are both pure states, $tr(\rho_1 \rho_2)$ is the inner product of $\rho_1$ and $\rho_2$. $\left|m\right>$ and $\left|n\right>$ are two different basis vectors, so $\left<m|n\right>=0$ and $\left<+_{mn}|-_{mn}\right>=0$ for $m \neq n$. A unitary transform keeps purity and inner products.

"$\Leftarrow$." Suppose $\mathcal{E}$ has the operator-sum representation $\mathcal{E}(\rho)=\sum_{i=1}E_i \rho E_i^{\dagger}$. According to condition (1),

\begin{align}
&tr\left[ \mathcal{E}(\left|m\right>\left<m\right|) \mathcal{E}(\left|n\right>\left<n\right|) \right] \notag\\
&= tr\left[\sum_{i,j} E_i\left|m\right>\left<m\right|E_i^{\dagger} E_j\left|n\right>\left<n\right|E_j^{\dagger} \right] \notag\\
&= \sum_{i,j} tr\left[E_i\left|m\right>\left<m\right|E_i^{\dagger} E_j\left|n\right>\left<n\right|E_j^{\dagger} \right] \notag\\
&= \sum_{i,j}\left|\left<m\right|E_i^{\dagger}E_j\left|n\right>\right|^2 = 0
\end{align}

Let $A_{ij} = E_i^{\dagger}E_j$. So $\forall i, j$ and $m \neq n$, $\left<m\right|A_{ij}\left|n\right> = 0$.  It means each $A_{ij}$ is diagonal on the basis $\{\left|1\right>\dots\left|d\right>\}$. According to condition (2), similarly, 

\begin{align}
&tr\left[ \mathcal{E}(\left|+_{mn}\right>\left<+_{mn}\right|) \mathcal{E}(\left|-_{mn}\right>\left<-_{mn}\right|) \right] \notag\\
&= \sum_{i,j}\left|\left<+_{mn}\right|A_{ij}\left|-_{mn}\right>\right|^2 = 0
\end{align}

\noindent It means $\forall i, j$, 

\begin{align}
\left<+_{mn}\right|A_{ij}\left|-_{mn}\right> = \frac{1}{2}(\left<m\right|+\left<n\right|)A_{ij}(\left|m\right>-\left|n\right>) = 0
\end{align}

\noindent so,

\begin{align}
&\frac{1}{2}(\left<m\right|A_{ij}\left|m\right> - \left<m\right|A_{ij}\left|n\right> + \left<n\right|A_{ij}\left|m\right> - \left<n\right|A_{ij}\left|n\right>) \notag\\
&= \frac{1}{2}(\left<m\right|A_{ij}\left|m\right> - \left<n\right|A_{ij}\left|n\right>) = 0
\end{align}

\noindent and thus for $m \neq n$, 

\begin{align}
\label{equ:mAmnAn}
\left<m\right|A_{ij}\left|m\right> = \left<n\right|A_{ij}\left|n\right>
\end{align}

Condition (2) selected $d-1$ combinations of $(m,n), m \neq n$ which form the edges of a connected graph with vertices $1,\dots,d$. Then from (\ref{equ:mAmnAn}), we can deduce that

\begin{align}
\left<1\right|A_{ij}\left|1\right> = \left<2\right|A_{ij}\left|2\right> = \dots = \left<d\right|A_{ij}\left|d\right>
\end{align}

\noindent It means the diagonal elements of $A_{ij}$ are equal, i.e., $\forall i,j, A_{ij} = E_i^\dagger E_j = k_{ij}I$, where $I$ is identity matrix.

For the case of $i=j$, $E_i^\dagger E_i=k_{ii}I$. Note that $(E_i^\dagger E_i)^\dagger = E_i^\dagger E_i$, which means $k_{ii}$ is a non-negative real number. Let $W_i=\frac{1}{\sqrt{k_{ii}}}E_i$, then $W_i^\dagger W_i=I$. So $W_i$ is a unitary matrix and $E_i = \sqrt{k_{ii}}W_i$. From $\sum_i E_i^{\dagger}E_i = I$ we can deduce that $\sum_i k_{ii} = 1$.

For the case of $i \neq j$, $E_i^\dagger E_j = \sqrt{k_{ii}k_{jj}}W_i^\dagger W_j=k_{ij}I$. Let $t_{ij} = \frac{k_{ij}}{\sqrt{k_{ii}k_{jj}}}$, then $W_i^\dagger W_j = t_{ij}I$ and thus $W_j=t_{ij}W_i$. Note that $W_i^\dagger W_j$ is a unitary matrix, so $|t_{ij}|=1$. Let $W=W_1, t_1 = 1$, and for $j \neq 1, t_j=t_{1j}$, then $\forall j, W_j = t_j W$ and $E_j = \sqrt{k_{jj}}t_j W$. In addition, $\sum_{j}\left|\sqrt{k_{jj}}t_j\right|^2 = \sum_j k_{jj} = 1$. According to Lemma 1, $\mathcal{E}$ is a unitary transform.
\end{proof}

\section{Proof of Proposition~\ref{prop:SelectT}}
\label{apdsec:ProofT}

\subsection{$s$ of Equivalence Checking}

Each round of the Swap Test can be seen as a Bernoulli experiment. Algorithm~\ref{Alg:EqCheck} execute three groups of Swap Test with the same number of rounds $s$. Denote the result of $i$-th round for Swap Test in lines 3, 4, and 5 as $x_i$, $y_i$, and $z_i$ respectively. Note that $x_i, y_i, z_i \in \{0,1\}$ and $s_1 = \sum_{i=1}^s x_i$, $s_2 = \sum_{i=1}^s y_i$, $s_{12} = \sum_{i=1}^s z_i$. Construct a variable $w_i$:

\begin{equation}
w_i = 2z_i - x_i - y_i
\end{equation}

\noindent Obviously $w_i \in [-2,2]$, and the mean value

\begin{align}
\bar{w} = \frac{1}{s}\sum_{i=1}^s w_i &= \frac{1}{s}\left( 2\sum_{i=1}^s z_i - \sum_{i=1}^s x_i - \sum_{i=1}^s y_i \right) \notag\\
&= \frac{2s_{12}-s_1-s_2}{s}
\end{align}

In Algorithm~\ref{Alg:EqCheck}, we estimate variable $r = \frac{2s_{12}-s_1-s_2}{s} = \bar{w}$. If two target programs are equivalent, whatever the input is, the distribution of every $x_i$, $y_i$, and $z_i$ are the same, and the expectation of $\bar{w}$:

\begin{equation}
E(\bar{w}) = \frac{1}{s}\left( 2\sum_{i=1}^s E(z_i) - \sum_{i=1}^s E(x_i) - \sum_{i=1}^s E(y_i) \right) = 0
\end{equation}

According to the \textit{Hoeffding's Inequality}~\cite{hoeffding1963},

\begin{equation}
Pr(\left|\bar{w} - E(\bar{w})\right| \geq \epsilon) \leq 2e^{\frac{-2\epsilon^2 s^2}{s \cdot (2-(-2))^2}}
\end{equation}

\noindent we obtain

\begin{equation}
Pr(|r|\geq\epsilon) \leq 2e^{-\frac{s\epsilon^2}{8}}
\end{equation}

\noindent $Pr(|r|\geq\epsilon)$ is the probability of failing in one test point. Note if two target programs are equivalent, it is the same for all test points. Let $P_2$ be the probability of passing in one test point, then

\begin{equation}
P_2 > 1 - 2e^{-\frac{s\epsilon^2}{8}}
\end{equation}

Running testing on $k$ test points, because if one test point fails, the testing result is "FAIL," the overall probability of failing is $1-P_2^k$, and it is also the probability of type II error. To control it not to exceed $\alpha_2$:

\begin{equation}
1-P_2^k \leq \alpha_2
\end{equation}

\noindent Then,

\begin{equation}
k \leq \frac{\ln(1-\alpha_2)}{\ln P_2} \leq \frac{\ln(1-\alpha_2)}{\ln(1-2e^{-\frac{s\epsilon^2}{8}})}
\end{equation}

\noindent Figure out $s$:

\begin{equation}
s \geq \frac{8}{\epsilon^2}\ln \frac{2}{1-(1-\alpha_2)^{\frac{1}{k}}}
\end{equation}

\subsection{$s$ of Unitarity Checking}

In Algorithm~\ref{Alg:UnCheck}, we estimate variable $r = 1-\frac{2s_1}{s}$. If the target program $\mathcal{P}$ is unitary, for any input $(\rho_1,\rho_2)$ selected by the algorithm, $tr(\mathcal{P}(\rho_1),\mathcal{P}(\rho_2)) = 0$. So $E(r) = 0$ and $E(s_1) = \frac{s}{2}$. According to the \textit{Chernoff's Bound}~\cite{chernoff1952}, 

%%【这里Chernoff's Bound需要一个文献引用？】

\begin{equation}
Pr(|s_1 - E(s_1)| \geq \epsilon E(s_1)) \leq 2^{-\frac{s\epsilon^2}{2}}
\end{equation}

\noindent That is

\begin{align}
Pr\left(|s_1 - \frac{s}{2}| \geq \frac{s\epsilon}{2}\right) \leq 2^{-\frac{s\epsilon^2}{2}} \notag\\
Pr(|r| \geq \epsilon) \leq 2^{-\frac{s\epsilon^2}{2}}
\end{align}

\noindent $Pr(|r|\geq\epsilon)$ is the probability of failing in one test point. Note if the target program is unitary, $Pr(|r|\geq\epsilon)$ is the same for all test points. Let $P_2$ be the probability of passing in one test point, then

\begin{equation}
P_2 > 1 - 2^{-\frac{s\epsilon^2}{2}}
\end{equation}

\noindent Similarly, limit the probability of failing:

\begin{equation}
1-P_2^k \leq \alpha_2
\end{equation}

\noindent Then,

\begin{equation}
k \leq \frac{\ln(1-\alpha_2)}{\ln P_2} \leq \frac{\ln(1-\alpha_2)}{\ln (1 - 2^{-\frac{s\epsilon^2}{2}})}
\end{equation}

\noindent Figure out $s$:

\begin{equation}
s \geq \frac{2}{\epsilon^2}\log_2 \frac{1}{1-(1-\alpha_2)^{\frac{1}{k}}}
\end{equation}

\subsection{Proof of Formula (\ref{equ:Tbound})}

\textit{The left less-than sign.} Note that $\forall a \in (0,1)$ and $x \in \mathbb{R}$,

\begin{equation}
a^x \geq x\ln a + 1
\end{equation}

\noindent It deduces that

\begin{equation}
\frac{1}{1-a^x} \geq \frac{1}{-x\ln a}
\end{equation}

\noindent Substitute $x$ with $1/k$ and $a$ with $1-\alpha_2$,

\begin{equation}
\frac{1}{1-(1-\alpha_2)^\frac{1}{k}} \geq \frac{k}{-\ln(1-\alpha_2)}
\end{equation}

\textit{The right less-than sign.} $k\geq 1$, so $1/k \in (0,1]$. Note that if $a\in(0,1)$ and $x\in[0,1]$,

\begin{equation}
a^x \leq (a-1)x+1
\end{equation}

\noindent It deduces that

\begin{equation}
\frac{1}{1-a^x} \leq \frac{1}{(1-a)x}
\end{equation}

\noindent Substitute $x$ with $1/k$ and $a$ with $1-\alpha_2$,

\begin{equation}
\frac{1}{1-(1-\alpha_2)^\frac{1}{k}} \leq \frac{k}{\alpha_2}
\end{equation}

\end{document}